%% file: mainPlain.tex
\documentclass[11pt,letterpaper]{article}
\usepackage[utf8]{inputenc}
\usepackage{fullpage}
\usepackage{amsmath, mathtools}
\usepackage{amsfonts} 
\usepackage[english]{babel}
\usepackage{color}
\usepackage[margin=1in]{geometry}
\usepackage{amsthm}
\usepackage{bm}
\usepackage{datetime}
\usepackage{enumitem}

\usepackage{amssymb}
\usepackage[compact]{titlesec}
\usepackage{nicefrac}
\usepackage{mathdots,url}
\usepackage{ifpdf,color}
\usepackage{graphicx}
\usepackage{subfigure}

\usepackage{thmtools, thm-restate}

\usepackage{dsfont} %% for nice vector 1   
\usepackage[linesnumbered,ruled,vlined]{algorithm2e} %%% change parameters for different looks of the pseudocode. 
\usepackage{xcolor}

\SetCommentSty{mycommfont}
\newtheorem{theorem}{Theorem}[section]  
\newtheorem{claim}[theorem]{Claim}
\newtheorem{proposition}[theorem]{Proposition}
\newtheorem{lemma}[theorem]{Lemma}

\theoremstyle{definition} 

\newtheorem{definition}[theorem]{Definition}
\newtheorem{remark}[theorem]{Remark}

\newtheorem{assumption}{Assumption}

\newenvironment{claimproof}[1]{\par\noindent\underline{Proof:}\space#1}{\hfill $\blacksquare$}

\newcommand{\R}{\ensuremath{\mathbb{R}}}

\newcommand{\1}{\ensuremath{\mathds{1}}}

\newcommand{\pr}[2]{\left\langle #1,#2\right\rangle}

\newcommand{\pOutbid}{\texttt{Outbid}}

\newcommand{\pRaisePrice}{\texttt{RaisePrice}}
\newcommand{\pFindNewPrices}{\texttt{Find\-New\-Prices}}
\newcommand{\pBEFindNewPrices}{\texttt{BE-Find\-New\-Prices}}

\DeclareMathOperator*{\mbb}{mbb}
\DeclareMathOperator*{\argmax}{arg\,max}
\DeclareMathOperator*{\nsw}{NSW}

\usepackage[disable,colorinlistoftodos]{todonotes}
\usepackage[colorlinks=true, allcolors=blue]{hyperref}
\newif\ifnotes\notesfalse

\usepackage[pagewise, displaymath, mathlines]{lineno}

 \SetKwFunction{findNewPrices}{FindNewPrices}
 \SetKwFunction{outbid}{Outbid}
 \SetKwFunction{increasePrice}{IncreasePrice}
 \SetKw{KwGoTo}{go to}

\title{Auction Algorithms for Market Equilibrium with Weak Gross Substitute Demands}
\date{\today}

\author{Jugal Garg\thanks{University of Illinois at Urbana-Champaign. Supported by the NSF grant CCF-1942321.}\\ \texttt{jugal@illinois.edu}
\and
Edin Husi\' c\thanks{Department of Mathematics, London School of Economics and Political Science. Supported by the European Research Council (ERC) under the European Union’s Horizon 2020 research and innovation programme (grant agreement ScaleOpt--757481).}\\ \texttt{e.husic@lse.ac.uk}
\and
L\' aszl\' o A. V\' egh\footnotemark[3]\\ \texttt{l.vegh@lse.ac.uk}
}
\date{}

\begin{document}

\maketitle
\thispagestyle{empty}
\begin{abstract}
We consider the Arrow--Debreu exchange market model under the assumption that the agents' demands satisfy the weak gross substitutes (WGS) property. We present a simple auction algorithm that obtains an approximate market equilibrium for WGS demands assuming the availability of a price update oracle. We exhibit specific implementations of such an oracle for WGS demands with bounded price elasticities and for Gale demand systems. 

As an  application of our result, we obtain an efficient algorithm to find an approximate spending-restricted market equilibrium for WGS demands, a model that has been recently introduced as a continuous relaxation of the Nash social welfare (NSW) problem. This leads to a polynomial-time constant factor approximation algorithm for the NSW problem with capped additive separable piecewise linear utility functions;  only a pseudopolynomial approximation algorithm was known for this setting previously.
\let\thefootnote\relax\footnote{The authors are grateful to the anonymous referees for numerous valuable suggestions that have helped to improve the presentation of the paper.}
\end{abstract}

\vfill
\tableofcontents
\bigskip

\newpage
\setcounter{page}{1}
\input{introduction}

\input{models}

\input{exchange} 

\input{FindingNewPrices} 

\input{galeEquilibrium}

\input{NSW}

\input{auctionAlgorithmSR}

\input{basplc}

\addtocontents{toc}{\protect\setcounter{tocdepth}{-1}}
\appendix
\input{runtime}
\input{CobbDouglasAgent}

%%%% References %%%%
\small
\bibliographystyle{abbrv}
\bibliography{references}

\end{document}

%% file: introduction.tex
%!TEX root = mainPlain.tex
\section{Introduction}
Market equilibrium is a fundamental model in mathematical economics to describe the balance between supply and demand.
The study of  market equilibria was pioneered by  Walras~\cite{walras1896elements} in 1874, and was developed in the 1950s by Arrow and Debreu \cite{arrow1954existence} and McKenzie \cite{Mckenzie1954}. In this paper, we focus on the classical exchange market setting, where a set of agents $A$ arrives at the market with initial endowments of infinitely divisible goods $G$. A market equilibrium comprises  prices for the goods and a fractional assignment between the goods and the agents. Prices and assignments form a \emph{market equilibrium} if {(a)} each agent receives a bundle of goods they prefer the most  at the given prices by spending their revenue  from selling their initial endowment, and {(b)} the market clears: the demand of each good meets its supply. 
A typical way to represent the preferences is by utility functions $u_i:\, \R_+^G\to\R_+$ for each agent $i\in A$; the demand of agent $i$ at prices $p$ and revenue $b_i$ is a bundle $x_i$ maximizing $u_i(x_i)$ subject to $\pr{p}{x_i}\le b_i$.
Classical works by Arrow and Debreu \cite{arrow1954existence} and McKenzie \cite{Mckenzie1954,Mckenzie1959} showed the existence of market equilibrium under mild assumptions, using Kakutani's fixed point theorem.

Equilibrium constitutes the ideal limit behavior of markets; existence proofs do not explain how such a limit can be attained.
Investigating market dynamics  has been an important topic since the early days: Walras \cite{walras1896elements} introduced the {\em t\^atonnement} process, a natural dynamics of supply and demand. This can be seen as a multi-round auction process, where in each round an auctioneer announces the current prices. At these prices, each agent submits their most preferred bundle of goods.  Prices are adjusted in light of these bids: prices of overdemanded goods are increased and prices of underdemanded goods are decreased.

Samuelson~\cite{Samuelson} formulated a continuous version of  t\^atonnement as a dynamical system. Works by Arrow and Hurwitz~\cite{ArrowH60}, Arrow, Block, and Hurwitz~\cite{ArrowBH} introduced the {\em weak gross substitutability (WGS)} property as a sufficient condition for convergence to an equilibrium. Agent's demands are said to be WGS if the demand for any good does not increase when its price increases while the rest of the prices remain unchanged.
Such a property gives a sound justification of the t\^atonnement price changes. 
However, it is a nontrivial requirement and there are important examples of demands 
that are not WGS. Scarf~\cite{Scarf60} showed that t\^atonnement may not converge for non-WGS demands.

\medskip
Classical market equilibrium models became subject of renewed interest in the optimization and theoretical computer science communities, starting from the 1991 paper by Megiddo and Papadimitriou \cite{Megiddo1991}. Formulating the existence of market equilibrium as a computational search problem raises intriguing questions. Many of these works treat market equilibrium as a centrally coordinated computational problem: can a {central authority} compute a market equilibrium given perfect information on all agents' utilities? Surprisingly, even in such a centrally coordinated setting, and even for some of the simplest non-WGS demands, computing  approximate equilibria turn out to be  complete problems for certain complexity classes. On the positive side, this line of investigation led to remarkable algorithmic developments for various market models.
See e.g., \cite{brainard2005compute,ChenDDT09,codenotti2004computation,CodenottiSVY06,DevanurPSV08,DuanM15,garg2017settling,GargV19,JainV06,VaziraniY11,Ye08}. 
For WGS utilities, the first polynomial-time computability of market equilibria was established by Codenotti, Pemmaraju, and Varadarajan \cite{codenotti2005polynomial}. 
A simple ascending-price algorithm using central coordination was given by Bei, Garg, and Hoefer~\cite{BeiGH19}.

In most market settings we cannot assume the level of central coordination needed for many of the above algorithms. 
In such markets, one has to investigate distributed mechanisms in a decentralized environment with limited coordination. 
Codenotti, McCune, and Varadarajan \cite{codenotti2005market} gave a simple discrete variant of the  t\^atonnement  algorithm that  converges to an approximate equilibrium for WGS utilities (see also \cite[Section 6.3]{nisan2007algorithmic}). This was followed by a number of papers providing t\^atonnement algorithms for various classes of utility functions and restricted models, some of them substantially weakening the need for central coordination among agents, see e.g., \cite{Avigdor2014,Cheung2019,Cheung2012,Cole2008,Fleischer2008}.

\medskip
{\em Auction algorithms} form an even simpler subclass of t\^atonnement-type algorithms. These algorithms are decentralized and require only local coordination between agents. Agent may take goods from others by outbidding them, i.e., offering  slightly higher prices. 
While  prices in t\^atonnement may increase as well as decrease,  prices in auction algorithms may only go up. For exchange market models, the first such algorithm was established for linear utilities---of the form $u(x)=\sum_{j\in G} v_{j}x_{j}$---by Garg and Kapoor \cite{GargK06} (see also \cite[Section 5.12]{nisan2007algorithmic}). The algorithm was later improved~\cite{garg2006price} and generalized to separable concave gross substitute utility functions~\cite{garg2004auction}, to a subclass of non-separable gross-substitutes called \emph{uniformly separable}~\cite{garg2007market}, and to a production model with linear production constraints and linear utilities~\cite{kapoor2007auction}. 

Auction algorithms have been widely used beyond  exchange markets, and studied in different contexts in optimization and economics. Bertsekas \cite{Bertsekas1981,Bertsekas1990} introduced auction algorithms for assignment and transportation problems. Closely related algorithms were introduced for markets with indivisible goods---further discussed in Section~\ref{sec:further}---by Kelso and Crawford \cite{Kelso1982}, and Demange, Gale, and Sotomayor \cite{Demange1986}. 

\subsection{Our contributions}
We present a new auction algorithm that computes an approximate market equilibrium in exchange markets for arbitrary WGS utilities, assuming a suitable oracle representation.  This settles an open question raised in~\cite{garg2007market}.  The result affirms the natural intuition that the WGS property should suffice for auction algorithms: A main invariant in auction algorithms is that at every price increase, the agents will still hold on to the goods they have purchased previously at lower prices. This property is almost identical to the definition of the WGS property; nevertheless, making an auction algorithm work for general WGS utilities requires new technical ideas. 

The previously mentioned auction algorithms operate with two prices for each good $j$, a lower price $p_j$ and a higher price $(1+\epsilon)p_j$. This technique was used   for linear~\cite{GargK06}, for separable~\cite{garg2006price}, and for  uniformly separable utilities~\cite{garg2004auction}.
However, this simple approach does not seem to be applicable for the general WGS case, and we need to  use a more fine-grained pricing approach. In our algorithm, each agent maintains individual prices for each good $j$ in the range  $[p_j,(1+\epsilon)p_j]$ for the `market price' $p_j$. The main invariant in our algorithm is that each agent maintains a subset of an optimal bundle with respect to these individual prices.  

Each agent updates their individual prices using a subroutine called \pFindNewPrices{}. The general algorithm in Section~\ref{section:exchange} relies on this subroutine and its running time for finding an $\varepsilon$-approximate equilibrium is $O\left(\frac{nmT_F}{\epsilon^2}\cdot \log
  \left(\frac{p_{\max}}{p_{\min}}\right)\right)$, where $T_F$ denotes a running time bound on  \pFindNewPrices{} and $p_{\max}$ and $p_{\min}$ are lower and upper bounds on the prices in an approximate equilibrium (Theorem~\ref{thm:running-oracle}).

 \paragraph{Demand systems}
 Our algorithm uses only local coordination between agents.
 However, agents should update their individual prices according to certain requirements. These are captured by 
 \pFindNewPrices{}; implementing this subroutine depends on the particular demand system. 

First, one needs to clarify how the preferences are represented in the model. 
WGS utilities in the literature are usually given in an explicit form such as CES (constant elasticity of substitution) or Cobb-Douglas utilities, see Section~\ref{sec:examples}. 
This is in contrast with the setting of markets with indivisible goods, where the common model is via a value or demand oracle \cite{Leme2017}, since direct preference elicitation, that is, the explicit description of the valuation function would be exponential. 
The class of continuous WGS functions also appears to be very rich, and hence an oracle approach seems more appropriate to devise algorithms for this class.

We model the agent preferences by \emph{demand oracles} (Definition~\ref{def:demand-oracle}). A demand oracle may be implemented by solving a utility maximizing convex problem, but may be of a different form. We discuss this in Section~\ref{sec:examples}, also exhibiting a class of WGS demand systems where our model is applicable, but do not appear to have a simple closed form representation.

A natural parametrization of WGS demand systems is by \emph{price elasticity} (Definition~\ref{def:flexible}) that bounds the change in the demands as a function of the price changes. In Section~\ref{section:newPrices-bounded},
we implement  \pFindNewPrices{} by a simple iterative application of the demand oracle for the case of bounded price elasticities.

We present additional implementations for the case when the price elasticity can be unbounded. 
Linear utilities constitute an important such class.
Lemma~\ref{lemma:linearPrices} gives a direct, linear time implementation of \pFindNewPrices{} for linear utilities.

In Section~\ref{section:gale} we consider Gale demand systems introduced  by Nesterov and Shikhman \cite{nesterov2018computation}. Such demand systems are given by  a convex program. Accordingly, we use a convex programming approach to implement \pFindNewPrices{}.

\paragraph{Spending restricted equilibrium and Nash social welfare}
\label{paragraph:NSW}
Our motivation for considering Gale demand systems comes from an application  to the \emph{Nash social welfare (NSW)} problem. In this problem, we need to allocate a set of indivisible goods to agents in order to maximize the geometric mean of their valuations.  

A useful relaxation of the NSW problem turns out to be a so-called \emph{spending-restricted (SR) equilibrium} under Gale demand systems. Spending-restricted equilibria were introduced by Cole and Gkatzelis \cite{cole2015approximating} as a key tool in finding the first constant-factor approximation algorithm for this problem with additive valuations. The same equilibrium concept was used in several other approximation algorithms for the NSW problem, see e.g.,~\cite{anari2018nash,ChaudhuryCGGHM18,cole2017convex,garg2018approximating}. 

The SR-equilibrium is a variant of the \emph{Fisher market model}, a special case of the exchange market model. In the Fisher model, the agents do not arrive with an initial endowment of goods but with a fixed budget to spend on the available set of goods. The SR-equilibrium differs in that the available amount of each good $j$ is influenced by the price $p_j$, namely, it is set as $\min\{1,1/p_j\}$. In other words, once the price of the good reaches 1, the seller will only sell an amount of total value 1.
Auction algorithms are well-suited for SR-equilibrium computation: once the price of a good goes above one, we can naturally decrease the total available amount of these goods within the auction framework.  

%In an extended version of this paper~\cite{garg2019auctionARXIV}, 
We design a polynomial time constant-factor approximation algorithm for the NSW problem under capped separable piecewise linear concave (SPLC) valuations by rounding an SR equilibrium under Gale demand systems to an approximately optimal solution.  Previous algorithm for this problem takes pseudopolynomial time~\cite{ChaudhuryCGGHM18}.  A key for this result is finding an approximate SR equilibrium in polynomial-time for which we use a modification of our auction algorithm.  
Interestingly, the capped separable piecewise linear concave valuations satisfy the WGS property under Gale demand systems, but not in the ``standard'' demand system setting~\eqref{prog:optimalBundle}. 
%For details we refer to~\cite{GargHV21,garg2019auctionARXIV}.

\subsection{Further related work}\label{sec:further}
\paragraph{Proportional response dynamics}
 \emph{Proportional response} is a distributed market mechanism introduced by Zhang~\cite{Zhang11} in the context of Fisher markets. 
In contrast with t\^atonnement and auctions, there is no direct price mechanism.
 In each round, agents bid on goods in proportional to the utility they receive from them in the previous round; the goods are then allocated in proportion of the agents' bids. Proportional response is known to converge to a market equilibrium in a variety of Fisher markets~\cite{BirnbaumDX11,CheungCT18,CheungHN19}, and some special cases of exchange markets~\cite{BranzeiDR19,BranzeiMN18,WuZ07}.

\paragraph{Markets with indivisible goods} Auction algorithms have been widely studied in the context of markets with {\em indivisible goods}. 
There are significant differences between the settings with divisible and indivisible goods. In the indivisible setting, equilibria are known to exist only in restricted settings. 
Kelso and Crawford \cite{Kelso1982} introduced (discrete) gross substitute utilities as a class where an equilibrium is guaranteed to exist, and  a simple auction algorithm can be used to 
find an approximate equilibrium. As shown by Gul and Stacchetti \cite{Gul1999,Gul2000}, the discrete gross substitutes property is in essence a necessary and sufficient condition for the existence of an equilibrium and for an auction algorithm to work. We refer the reader to the survey by Paes Leme \cite{Leme2017} on the role of gross substitute utilities in markets with indivisible goods, and their connections to  discrete convex analysis.

Whereas the definitions of discrete gross substitutes and continuous WGS utilities are very similar, there does not appear to be a direct connection between these notions. The main difference is in the utility concepts: for indivisible markets, the standard model is to maximize the valuation minus the price of the set at given prices, whereas  standard divisible market models operate with {\em fiat money}: the prices appear via the budget constraints but not in the utility value. Still, our result can be interpreted as the continuous analogue of the strong link between auction algorithms and the gross substitutes property for markets with indivisible goods: we show that auction algorithms are applicable for the entire class of WGS utilities for markets with divisible goods. We suspect that the converse should also be true, namely, that the applicability of auction algorithms should be limited to WGS utilities. In contrast, t\^atonnement algorithms have been successfully applied beyond the WGS class, see e.g., \cite{Cheung2019,Cheung2012,Fleischer2008}.

\paragraph{Graphical exchange economies} Subsequently to the  preliminary version of this work~\cite{GargHV21}, Andrade, Frongillo, Gorokhovsky, and Srinivasan~\cite{andrade2021graphical} studied graphical exchange markets with resale. Here, agents may only trade with their neighbors in a graph. 
They show the existence of such an equilibrium and give an auction algorithm for finding an approximate market equilibrium in such markets assuming that agents have WGS demands.

\medskip
\paragraph{Organization} The rest of the paper is structured as follows. Sections~\ref{section:models}--\ref{sec:findnewprices} deal with the auction algorithm for exchange markets under WGS demands. In particular, Section~\ref{section:models} defines the exchange market model and provides examples of WGS demand systems. Section~\ref{section:exchange} presents the auction algorithm for exchange markets. Section~\ref{sec:findnewprices} presents different ways of implementing \pFindNewPrices{}--the key subroutine of the algorithm.

Sections~\ref{sec:fisher}--\ref{section:BASPLC} present the modified auction algorithm for finding SR-equilibria in Fisher markets, and how this can be used for designing a polynomial-time constant-factor approximation algorithms for the NSW problem under capped SPLC utilities. In particular, Section~\ref{sec:fisher} defines the NSW problem and its connection with SR-equilibria.  Section~\ref{section:SR} presents the modified auction algorithm for SR-equilibria. Section~\ref{section:BASPLC} describes the approximation algorithm for the NSW problem.
\smallskip

A preliminary version of this paper appeared in~\cite{GargHV21}.
%%% Local Variables:
%%% mode: latex
%%% TeX-master: "main"
%%% End:

%% file: models.tex
%!TEX root = mainPlain.tex
\section{The exchange market and demand systems}\label{section:models}
We use $\R_+$ for the nonnegative reals, and for a positive integer $k$, let $[k]=\{1,2,\ldots,k\}$. 
We consider a market with a set of agents $A=[n]$ and a set of divisible goods
$G=[m]$.
Each agent $i\in [n]$ arrives at the market with an initial endowment of
goods $e^{(i)} \in \R_+ ^ m$. We let $e = \sum_{i = 1}^n e^{(i)}$ denote the total amount of the goods. We assume $e_j>0$ for each $j\in[m]$.
A \emph{bundle} $x$ is a non-negative vector $x \in \R^m_+$.
We say that a bundle of goods $y\in\R_+^m$ \emph{dominates} the bundle $x\in\R_+^m$ if $x\le y$.

Given a non-negative price vector $p \in \R^m_+$, the budget of agent $i$ at prices $p$ is defined as $b_i(p)=\pr{p}{e^{(i)}}$; we simply write $b_i$ if the prices are clear from the context.
It follows that $\pr{p}{e}= \sum_{i=1}^n \pr{p}{e^{(i)}} = \sum_{i=1}^n b_i$.

We specify the markets via {\em demand systems}. 
 A {\em demand system} is a function $D : \R^{m+1}_+ \to 2^{\R^m_+}$; $D(p,b)$ denotes the set of preferred bundles of an agent at prices $p$ that are affordable within budget $b$. Here $2^{\R^m_+}$ denotes the family of all subsets of $\R^m_+$. Bundles in $D(p,b)$ are called \emph{optimal bundles} or \emph{demand bundles} at prices $p$ and budget $b$.
The demand system is \emph{simple}
if $|D(p,b)| =1$ for all $(p, b) \in \R^{m+1}$; for such demand systems, we will also use $D(p,b)$ to
denote this single bundle. 
We make two assumptions on the demand systems.

\begin{assumption}[Scale invariance]\label{assumption1}
For every agent $i\in[n]$, $(p,b)\in\R^{m+1}_+$, and  $\alpha > 0$,
$D_{i}(p, b) = D_{i}(\alpha p, \alpha b)$ holds.
\end{assumption}
That is, we require that the demand is homogeneous of degree $0$; informally, the demand does not depend on the currency.
This is a standard assumption in microeconomics, see e.g.,~\cite{arrow1958stability,devanur2004spending,Eisenberg1961,matsuyama2017beyond}.

\begin{assumption}[Non-satiation]
\label{non-satiation}
For every agent $i\in[n]$, $(p,b)\in\R^{m+1}_+$,  and every
$x\in D_i(p,b)$, we have $\pr{p}{x}=b$.
\end{assumption}
That is, in every
optimal bundle the agents must fully spend their budgets. This is a
standard assumption for exchange markets as it  is necessary for
the fundamental theorems of welfare economics (see e.g., \cite[Chapter
16]{Mas1995book}).\footnote{We note that this assumption can be replaced 
by a weaker one in the case of Fisher markets, see~\cite{GargHV21,garg2019auctionARXIV}.}

\medskip
A common way to define 
demand systems is by utility functions. By a \emph{utility function} we mean a function $u : \R^m_+ \to \R_+$ that  is concave, continuous, non-decreasing, and $u(0)=0$.  The corresponding demand system is 
\begin{equation}\label{prog:optimalBundle}
\begin{aligned}
D^u(p,b):=\arg\max_{x \in \R^m_+} \left\{u(x): \pr{p}{x}
\le b\right\}\, .
\end{aligned}
\end{equation}
An important example is the \emph{linear demand system}  defined by a  linear utility function $u(x)=\pr{v}{x}$
for $v\in\R^m_+$. The corresponding demand system $D^u(p,b)$ is the set of all fractional assignment of goods maximizing $v_j/p_j$ (bang-per-buck) with a total price $b$. Thus, this demand system is not simple.

For demand systems in the form \eqref{prog:optimalBundle}, Assumption~\ref{assumption1} is immediate, and Assumption~\ref{non-satiation} holds if $u(x)$ is strictly monotone increasing. The demand system is simple if $u(x)$ is strictly convex.

\subsection{Exact and approximate market equilibria}
\begin{definition}[Market equilibrium]\label{def:market-eq}
Consider an exchange market with a set $A=[n]$ of agents, a set $G=[m]$ of goods, initial  endowments $e^{(i)} \in \R_+ ^ m$, and $e = \sum_{i = 1}^n e^{(i)}$. 
Let $D_i(p,b)$ denote the demand system of agent $i\in A$.
The prices $p\in \R^m_+$ and bundles $x^{(i)} \in \R^m_+$
 form a \emph{market equilibrium} if
 \begin{enumerate}[label=(\roman*)]
 \item  $x^{(i)}\in D_i(p,\pr{p}{e^{(i)}})$ for all $i\in A$, and
\item  $\sum_{i=1}^n x^{(i)}_j\le e_j$, with equality whenever $p_j>0$, for all $j\in G$.
\end{enumerate}
\end{definition}

That is, $p$ and the optimal bundles $x^{(i)}$ form an equilibrium if no good is overdemanded and goods at a positive price are fully sold. Note that this implies that every agent fully spends their budget.

We relax this to the following notion of $\epsilon$-approximate equilibrium:

\begin{definition}[Approximate market equilibrium]\label{def:approx-eq} Consider the same setting as in Definition~\ref{def:market-eq}.
For $\epsilon\ge 0$, the prices $p\in \R_+^m$ and 
bundles $x^{(i)} \in \R^m_+$ 
 form an \emph{$\epsilon$-approximate market equilibrium} if
\begin{enumerate}[label=(\roman*)]
\item  \label{def:subset} $x^{(i)}\le z^{(i)}$ for some $z^{(i)}\in
  D_i(p^{(i)},\pr{p}{e^{(i)}})$, where $p\le p^{(i)}\le (1+\epsilon)
  p$, 
\item\label{def:no-over} $\sum_{i=1}^n x^{(i)}_j\le e_j$, and
\item \label{def:leftover} $\pr{p}{e-\sum_{i=1}^n x^{(i)}}\le
  \epsilon \pr{p}{e}$.
\end{enumerate}
\end{definition}

That is, every agent owns a subset of their optimal bundle at prices
that are within a factor $(1+\epsilon)$ from $p$, and all goods are
nearly sold: the value of the unsold goods is at most an $\epsilon$
fraction of the total value of the goods. The total value of the goods ``taken
away'' from the near-optimal bundles of the agents is $\sum_{i=1}^n
\pr{p}{z^{(i)}-x^{(i)}}$. 
Parts \ref{def:subset} and
\ref{def:leftover}, together with the fact that $\pr{p^{(i)}}{z^{(i)}} \le \pr{p}{e^{(i)}}$ for all $i$, 
imply that this amount is at most $\epsilon \pr{p}{e}$. In particular, $\epsilon=0$ corresponds to an exact market equilibrium as in Definition~\ref{def:market-eq}.

Condition \ref{def:subset} can be seen as a natural extension
of the corresponding approximate optimality conditions in {previous auction algorithms}~\cite{GargK06,garg2007market,garg2004auction}. 
For linear utilities, {Garg and Kapoor}~\cite{GargK06} require the approximate maximum bang-per-buck condition $v_{ij}/p_{j} \le (1+\epsilon) v_{ik} / p_k$ for any agent $i$, goods $j$ and $k$ such that $x_{ik}>0$. In other words, the goods purchased by agent $i$ according to this definition are maximum bang-per-buck with respect to some prices $p^{(i)}$ such that
$p\le p^{(i)}\le (1+\epsilon)p$.

Condition \ref{def:leftover} corresponds to the definition of approximate
equilibrium in~\cite{devanur2003improved}
and~\cite{ghiyasvand2012simple}.
This notion is weaker than the ones used in
\cite{GargK06,garg2007market,garg2004auction}. The most important
difference is that the latter papers guarantee that each agent
recovers approximately their optimal utility. Such a property could be
achieved by strengthening the bound in \ref{def:leftover} from
$\epsilon \pr{p}{e}$ to $\epsilon p_{\min} e_{\min}$, where $p_{\min}$
is the minimum price and $e_{\min}$ is the smallest total fractional
amount in the initial endowment of any agent. However, this would come
at the expense of substantially worse running time guarantees in our
algorithmic framework.

\medskip
An important special case of exchange markets are \emph{Fisher markets}, where $e^{(i)}=\frac{b_i}{\sum_{i=1}^n b_i} e$ for each $i\in [n]$, where $b_i>0$. That is, the initial endowments include every good in the same proportion. By appropriately scaling the prices, we can interpret the $b_i$'s as fixed budgets, and an exchange market equilibrium can be written as follows.

\begin{definition}[Fisher market equilibrium]\label{def:fisher-market-eq}
Consider a Fisher market with a set $A=[n]$ of agents, a set $G=[m]$ of goods, and budgets $b_i>0$, $i\in[n]$. 
Let $D_i(p,b)$ denote the demand system of agent $i\in A$.
The prices $p\in \R^m_+$ and bundles $x^{(i)} \in \R^m_+$
 form a \emph{Fisher market equilibrium} if
 \begin{enumerate}[label=(\roman*)]
 \item  $x^{(i)}\in D_i(p,b_i)$ for all $i\in A$, and
\item  $\sum_{i=1}^n x^{(i)}_j\le e_j$, with equality whenever $p_j>0$, for all $j\in G$.
\end{enumerate}
\end{definition}

\subsection{The weak gross substitutes property}
\label{sec:examples}
We next introduce the class of demand systems investigated in this paper.
\begin{definition}\label{def:WGS}
The demand system $D(p,b)$ is a   \emph{weak gross substitutes (WGS)  demand system} if for any
 $(p, b) \in \R_+^{m+1}$, any $x \in D(p, b)$, and 
 any $p' \ge p$ and $b'\ge b$, there exists $y \in D(p', b')$ such that $y_j \ge x_j$ whenever $p'_j=p_j$.

Further, we say that the utility function $u : \R^m_+ \to \R_+$ satisfies the WGS property if the corresponding demand system $D^u(p,b)$ as in \eqref{prog:optimalBundle} is a WGS demand system.
\end{definition}

%In the  t\^atonnement literature, the weak gross substitutes property is usually defined with respect to the {\em aggregate} excess demand function of all agents. Similarly to the previous auction algorithms \cite{garg2007market,garg2004auction}, we require the WGS property for the demand system of every individual agent.\lnote{removed most of this discussion; do we need to say anything more?}\enote{It is enough. Maybe even remove it?} \jnote{yes, let's remove it.}
We use an oracle model to represent the demand systems. We require access to the allocations guaranteed by Definition~\ref{def:WGS}.

\begin{definition}[Demand oracle]\label{def:demand-oracle}
For a WGS demand system $D(p,b)$, a {\em WGS demand oracle} requires in the input two vectors
$(p, b)$, $(p',b') \in \R_+^{m+1}$ such that $(p',b')\ge (p,b)$, 
and a vector $x\in D(p,b)$. The oracle outputs a vector $y\in D(p',b')$ such that 
that $y_j \ge x_j$ whenever $p'_j=p_j$.
\end{definition}

The complex form of the definition is due
to the possible non-uniqueness of demand bundles.
For simple demand systems, it suffices to specify $(p', b')\in \R_+^{m+1}$ in the input;  the output is the unique vector $D(p',b')$.

Consider a demand system $D^u(p,b)$ as in \eqref{prog:optimalBundle} for a utility function $u : \R^m_+ \to \R_+$. If this is not a simple demand system, we can implement the demand oracle by adding the constraints $y_i\ge x_i$ for every $i$ with $p'_i=p_i$ to the convex optimization problem in~\eqref{prog:optimalBundle}.

\paragraph{Examples of WGS utilities} 
Some classical examples in the literature are as follows.
\begin{itemize}
\item As previously mentioned, the {\em linear utility function} is given by $u(x) =
 \pr{v}{x}$ for  $v \in \R^m_+$. 
\item The {\em Cobb-Douglas} utility function is specified by parameters $\alpha\in\R_+^m$, $\sum_{j=1}^m \alpha_j = 1$ as
\[
u(x) :=
  \prod_{j=1}^m x_{j}^{\alpha_{j}}\, .
  \]
This is a simple demand system with $x=D^u(p,b)$ such that
$\displaystyle x_j = b\alpha_j/p_j$ for all $j\in[m]$.
\item The {\em constant elasticity of substitution (CES)} utility function
  is specified by parameters $\beta\in\R^m_+$ such that $\sum_{j=1}^m \beta_j =
  1$, and $\sigma\in\R_+$ as  
  \[
  u(x) := \left( \sum_{j=1}^m
    \beta_j^{\frac{1}{\sigma}} x_{j}^{\frac{\sigma -1}{\sigma}}
  \right)^{\frac{\sigma}{\sigma -1}}\, ,
  \]
   This is also a simple demand system with $x=D^u(p,b)$ such that $\displaystyle x_j = \frac{ \beta_j p_j^{- \sigma} b}{\sum_{k=1}^m \beta_k p_k^{1-\sigma}}$ for all $j\in[m]$.
	The CES demand system satisfies the WGS property if and only if $\sigma > 1$. 
\item The {\em nested CES} utility function is defined recursively (see~\cite{JainV06} for more details). Any CES function is a nested CES function. If $g, h_1, \dots, h_t$ are nested CES functions, then $f(x)= \max g(h_1(x^1), \dots, h_t(x^t))$ over all $x^1, \dots, x^t$ such that $\sum_{k=1}^t x^k = x$, is a nested CES function. In a well-studied special case, each good $j$ can only be used in at most one of the $h_i$'s; see e.g.,~\cite{Keller76}.
\end{itemize}

\paragraph{Convex combinations of demand systems}
Given two WGS utility functions $u$ and $u'$, the demand system
corresponding to their sum $u+u'$ may not be WGS. In contrast, taking convex combinations of simple WGS demand systems retains this property; the following proposition is easy to verify.
\begin{proposition}
Let $D(p,b)$ and $D'(p,b)$ be two simple WGS demand systems and  $0\le \lambda\le 1$.
Let us define the demand system $D''=\lambda
D+(1-\lambda) D'$ by
\[
D''(p,b):=D(p, \lambda b) + D'(p,(1-\lambda) b)\, .
\]
Then, $D''$ is a simple WGS demand system.
\end{proposition} 
This enables the
construction of some interesting demand systems. For example, Matsuyama and Ushchev~\cite{matsuyama2017beyond} consider hybrids of CES and
Cobb-Douglas demands, where the demand system can be given as 
 \[
 x_j = \frac{b}{p_j}\left[\lambda \alpha_j + (1-\lambda)
 \frac{ \beta_j p_j^{1- \sigma}}{\sum_k \beta_k p_k^{1-\sigma}}\right]
 \, ,\] 

 for $\beta\in\R^m_+$, $\sum_{j=1}^m \beta_j =
  1$, $\sigma>1$,  $\alpha\in\R_+^m$, $\sum_{j=1}^m \alpha_j = 1$, and $0\le \lambda\le 1$.\footnote{We note that this demand function does not seem to correspond to a nested CES utility function.} 
 
 Note that if $D=D^u$ and $D'=D^{u'}$ for some concave utility
 functions $u$ and $u'$, the demand system $\lambda
D+\lambda' D'$ will in general not correspond to the utility function 
$\lambda u+\lambda' u'$. It is not clear whether one can
explicitly write utility functions corresponding to such convex combinations.

Using a demand oracle model, our algorithm is applicable to  convex combination of simple demand oracles.

\paragraph{Separable and uniformly separable WGS utility functions}
The auction algorithm for linear utilities~\cite{GargK06}
was later extended to separable WGS utility functions~\cite{garg2004auction}, that is,
$u=\sum_{j\in G} u_j$ where each $u_j$ is a WGS utility function
depending only on good $j$. This model was further generalized  to {\em uniformly separable} WGS utility
functions~\cite{garg2007market}, 
that is, $\frac{\partial u(x)}{\partial x_j}= f_{j}(x_j)
g(x)$, where each $f_j$ is a strictly decreasing function. 
This class already includes CES and Cobb-Douglas utilities; however,
it does not appear to extend to demand systems obtained as their convex
combinations, where even the explicit form of the utility function is
unclear. Further, the running time bound stated in
\cite{garg2007market} is unbounded for the CES and Cobb-Douglas cases; 
see Appendix~\ref{section:running-times} for further discussion.

\subsection{Price elasticity of demands} 
A commonly studied property of demand systems is 
{\em price elasticity}.
For simple demand systems that are differentiable, the usual definition of the price
 elasticity of good $j$  with respect to the price of good $k$ is
$e_{j,k}=\partial \log D_j(p,b)/\partial \log p_k$, where $D_j(p,b)$
is the unique demand for good $j$ at prices $p$ and budget
$b$. The WGS property guarantees that $e_{j,k}\ge 0$ if $j\neq k$, and
consequently, $e_{k,k}\le 0$. 

The following definition  does not assume simplicity or differentiability of the demand system. It corresponds to 
$e_{k,k}\ge -f$ for all $k \in [m]$, in the above case.
\begin{definition}\label{def:flexible}
Consider a WGS demand system $D(p,b)$. For some $f>0$, we say that the {\em elasticity of
  $D(p,b)$ is at least $-f$} if the following holds.
For any $(p,b)\in \R^{m+1}_+$ and $x\in D(p,b)$, $j\in [m]$ and $\mu\ge 1$, let us define
\begin{equation}\label{eq:price-inc}
p'_j=\begin{cases} \mu p_k&\mbox{if }k=j\, ,\\
p_k &\mbox{otherwise}.\end{cases}
\end{equation}
Then, there exists a bundle $x' \in D(p', b)$ such that $x'_j \ge x_j/\mu^f$ and $x'_k\ge x_k$ for every $k\neq j$.\

\end{definition}

It is easy to check that the linear demand systems do not satisfy this property for any finite $f$, as the demand for a good may drop to zero by an arbitrary small price increase. We include the proof of the following well-known statement to illustrate this concept.

\begin{lemma}\label{lem:CES-flex}
The Cobb-Douglas demand system has elasticity at
least $-1$, and 
the CES demand system with parameter $\sigma>1$  has elasticity at
least  $-\sigma$.
\end{lemma}

\begin{proof}
The optimal bundle for a Cobb-Douglas utility function is $x=D(p,b)$ with  $\displaystyle x_\ell = b\alpha_\ell/p_\ell$ for $\ell\in[m]$. 
Increasing the price  of a good by a factor $\mu\ge 1$ corresponds to a decrease in the demand by the same factor. Thus, the elasticity is at least $-1$.

\smallskip
The optimal bundle for CES utilities  is  $x=D(b,p)$ with   $\displaystyle x_\ell = \frac{ \beta_\ell p_\ell^{- \sigma} b}{\sum_{k=1}^m \beta_k p_k^{1-\sigma}}$ for  $\ell\in[m]$.
Select any  good $j\in [m]$ and $\mu\ge 1$, and let $p'$ be defined as in \eqref{eq:price-inc}. Let $x'=D(p',b)$ denote the optimal bundle. 
Using $\sigma>1$, we get 
\begin{align*}
x'_j = \frac{ \beta_j \mu^{-\sigma}p_j^{- \sigma} b}{\sum_{k\neq j} \beta_k p_k^{1-\sigma} + \beta_j \mu^{1-\sigma}p_k^{1-\sigma}} 
= \frac{ \beta_j p_j^{- \sigma} b}{\mu^{\sigma}\sum_{k\neq j} \beta_k p_k^{1-\sigma} + \beta_j \mu p_k^{1-\sigma}}
> \frac{ \beta_j p_j^{- \sigma} b}{\mu^{\sigma} \sum_{k} \beta_k p_k^{1-\sigma}} = \frac{x_j}{\mu^{\sigma}}  \, ,
\end{align*}
verifying that the CES demand system has elasticity at least $-\sigma$.
\end{proof}

Our next lemma allows us to derive price elasticity bounds for convex combinations of simple demand systems.

\begin{lemma}\label{lemma:convexComb}
  Let $D$ and $D'$ be simple demand systems with elasticity at least $-f$ and $-f'$, respectively. Let $0\le \lambda\le 1$.
  Then the demand system $\lambda D + (1-\lambda)D'$ has elasticity at least $\min\{-f, -f'\}$.
\end{lemma}
\begin{proof}
Let $D''=\lambda D + (1-\lambda)D'$ and  $f'' = \max\{f, f'\}$.
Let $(p,b)\in\R^{m+1}_+$,
$x = D(p, \lambda b)$ and $x' = D''(p,(1-\lambda) b)$. 
Then, $x''=x+x' = D'(p, b)$.

Let $j\in[m]$, $\mu \ge 1$ and define $p'$ as in~\eqref{eq:price-inc}.
As the elasticity of $D$ is at least $-f$,  for $y = D(p', \lambda b)$ we have $y_j \ge {x_j}/{\mu^{f}} \ge x_j/{\mu^{f''}}$.
Analogously, for $y' = D'(p', (1-\lambda) b)$ we have $y'_j \ge {x'_j}/{\mu^{f''}}$.
Thus, $y_j+y'_j \ge (x_j+x'_j)/{\mu^{f''}}$.
Since $y+y'= D''(p, b)$, 
we conclude that the elasticity of $D''$ is at least $-f''$.  
\end{proof}

\subsection{Gale demand systems}
\label{sec:GaleDemandSystems}
Recall that for a utility function $u:\R^m_+ \to \R_+$,
we can obtain demand systems from utilities using the convex program \eqref{prog:optimalBundle} that maximizes the utility subject to the budget constraint.

Fisher market equilibria can be formulated by the well-known  Eisenberg--Gale convex program  \cite{Eisenberg1959} for many important cases:
\begin{equation}\label{eq:EG}
\max \sum_{i=1}^n b_i \log u_i(x_i)\,\quad \mathrm{s.t.}\quad \sum_{i=1}^n x_i\le e\, .
\end{equation}
Eisenberg~\cite{Eisenberg1961} showed that the optimal solutions to this program, together with the prices corresponding to the optimal Lagrangian multipliers, form a Fisher market equilibrium if the utilities are homogeneous of degree one---that is, $u_i(\alpha x)=\alpha u_i(x)$ for every $x\in\R^m$ and $\alpha>0$, $i\in [n]$.
This class includes many important examples such as linear, Cobb--Douglas, and CES utilities.

Nevertheless, solutions to \eqref{eq:EG} may not correspond to a Fisher market equilibrium in general. 
Nesterov and Shikhman~\cite{nesterov2018computation} showed, using Lagrangian duality, that the optimal solutions to \eqref{eq:EG} always form a Fisher market equilibrium for 
{\em Gale demand systems} defined as:
\begin{equation}\label{galeObjective}
G^u (p, b) := \argmax_{x \in \R^m_+} \, b \log u(x) - \pr{p}{x}  \, .
\end{equation}
The following connection explains why the Eisenberg--Gale program can be used for demand systems of the form  \eqref{prog:optimalBundle} for homogeneous degree one utilities. 
The proof follows easily from Lagrangian duality and Euler's homogeneous function theorem~\cite[Section 6.2]{nisan2007algorithmic}.  
\begin{lemma}\label{lem:gale-equiv}
Let $u$ be a utility function that is homogeneous of degree one and differentiable. Then, for any $(p,b)\in\R^{m+1}_+$, the optimal solutions to the systems \eqref{prog:optimalBundle} and 
\eqref{galeObjective} coincide.
\end{lemma}

Nesterov and Shikhman~\cite{nesterov2018computation} study Gale equilibria (equilibrium under Gale demand systems) as well as the
more general concept of Fisher-Gale equilibria; they also give a
t\^atonnement type algorithm for finding such an equilibrium.

\smallskip
Gale demand systems have already been used for designing approximation algorithms for the NSW problem using its connection with SR-equilibria under Gale demand systems. We will adopt our auction algorithm for SR-equilibria under WGS Gale demand systems, and use it to present a constant-factor approximation algorithm for the NSW problem under capped-SPLC utilities. We elaborate further in Section~\ref{sec:fisher}.
%%% Local Variables:
%%% mode: latex
%%% TeX-master: "main"
%%% End:

%% file: exchange.tex
%!TEX root = mainPlain.tex
\section{The auction algorithm}
\label{section:exchange}
Algorithm~\ref{exchangeMarkets} describes the auction algorithm. It outputs a $4\epsilon$-approximate market equilibrium for an accuracy  parameter $0<\epsilon<0.25$ specified in the input.
We use the notation $e$, $e^{(i)}$, $D_i$ as in Definitions~\ref{def:market-eq} and \ref{def:approx-eq}. We introduce some notation and formulate key invariants.

\begin{enumerate}[label=(\alph*)]
\item\label{prop:price-inc}
We maintain a price vector $p$ called  {\em market prices}, 
initialized as $p_j=1$ for all $j\in [m]$.\footnote{Recall from
Assumption~\ref{assumption1} that if there exist market clearing prices that are strictly positive, we can also assume that these prices are at least 1.
Even though there might be goods priced at 0 in an equilibrium, we can always find an $\epsilon$-approximate market equilibrium where all prices are positive.}
Prices may only increase, and remain integer powers of  $(1+\epsilon)$.
\item\label{prop:fully-sold} No good is oversold, i.e., at most $e_j$ amount is sold of each good $j=1,2,\ldots,m$. The market price is $p_j=1$ for every good $j$ that is not fully sold.
\item\label{prop:individual} The budget of agent $i$ is $b_i=\pr{p}{e^{(i)}}$. 
  Every agent $i\in [n]$ maintains \emph{individual prices} $p^{(i)}\in\R_+^m$
  such that $p\le p^{(i)}\le (1+\epsilon)p$. We let
  \[
  L_i  := \{j \in [m] : p^{(i)}_j < (1+\epsilon)p_j \}\, \quad \mbox{and}\quad   H_i := [m] \setminus L_i\, .
  \]
\item\label{prop:dominate} Every agent $i\in[n]$
   owns a bundle of goods
  $c^{(i)}\in\R_+^m$ that is dominated by a  bundle $x^{(i)} \in D_i(p^{(i)},b_i)$, i.e., an optimal bundle with respect to the individual prices $p^{(i)}$ and the budget $b_i$.
  We call $x^{(i)}$ the \emph{desired bundle}.
  \item\label{prop:pay} For the amount $c^{(i)}_j$ of good $j$, agent $i$ pays $p_j$ if $j\in L_i$ and $(1+\epsilon)p_j$ if $j\in H_i$.\footnote{This is in contrast with \cite{GargK06} and the other previous
auction algorithms where $i$ may pay $p_j$ for some amount of good $j$
and $(1+\epsilon)p_j$ for another amount.} The \emph{surplus} of agent $i$ is 
\[
s_i:=b_i-\sum_{j \in L_i} c^{(i)}_j  p_j - (1+\epsilon)\sum_{j\in
      H_i}   c^{(i)}_j  p_j\, .
\]
\end{enumerate}
Before giving an overview of the algorithm, we formulate the termination condition.
\begin{lemma}\label{lem:term}
Assume that \ref{prop:price-inc}--\ref{prop:pay} hold as above. Then $s_i\ge 0$ for all $i\in [n]$. Moreover, if 
\[
\sum_{i=1}^n s_i\le 3\epsilon\pr{p}{e}\]
 for the input bundle $e\in\R^m_+$, then the market prices $p$ and allocations $c^{(i)}$, $i=1,2,\ldots,n$ form a $4\epsilon$-approximate market equilibrium.
\end{lemma}
\begin{proof} 
Let $c^{(i)}$ be the bundle of goods agent $i$ owns. By invariant~\ref{prop:dominate}, 
there exists a desired bundle $x^{(i)}$ dominating $c^{(i)}$.
The bundle $x^{(i)}$ is affordable for $i$ at prices $p^{(i)}$, and thus by invariants~\ref{prop:dominate} and~\ref{prop:pay} 
the same bundle is affordable for $i$ at prices $p_j$ for $j\in L_i$ and $(1+\epsilon)p_j$ for $j \in H_i$. 
It follows that $i$ can afford $c^{(i)}$ at prices $p_j$ for $j\in L_i$ and $(1+\epsilon)p_j$ for $j \in H_i$. Hence, $s_i \ge 0$.

\smallskip
Condition  \ref{def:subset} in Definition~\ref{def:approx-eq} is immediate from invariants \ref{prop:individual} and \ref{prop:dominate}, and condition~\ref{def:no-over} follows from \ref{prop:fully-sold}. It is left to verify condition \ref{def:leftover}. We can write
\[
\begin{aligned}
\pr{p}{e-\sum_{i=1}^n c^{(i)}}&=\sum_{i=1}^n \pr{p}{e^{(i)}}-\sum_{i=1}^n \pr{p}{c^{(i)}}=\sum_{i=1}^n \left( b_i - \pr{p}{c^{(i)}}\right)\\
&=\sum_{i=1}^n\left( s_i +\epsilon \sum_{j\in H_i}   c^{(i)}_j  p_j\right)=\sum_{i=1}^n  s_i +\epsilon  \pr{p}{e}\le 4\epsilon  \pr{p}{e}\, . 
%\qedhere
\end{aligned}\]
\end{proof}

We now give an overview of the algorithm. 
The individual prices  $p^{(i)}$ are updated by the key subroutine \pFindNewPrices{} that outputs prices and bundles as specified below. In Section~\ref{sec:findnewprices}, we provide implementations for different classes of demand systems.

\medskip
\begin{center} \fbox{\begin{minipage}{0.87\textwidth} \noindent
{\sf Subroutine} {\pFindNewPrices{}} 

\smallskip
\begin{tabular}{ r l  }
{\bf Input:} & Agent $i\in[n]$, market prices $p\in\R^m_+$, individual prices $p^{(i)}\in\R^m_+$ \\ 
            & such that $p\le p^{(i)}\le (1+\epsilon)p$, budget $b_i\in\R_+$, and bundle $c^{(i)}\in\R_+^m$.\\

{\bf Output:} & Prices $\tilde p\in\R^m_+$ and bundle $y\in\R^m_+$ such that \\
&\begin{minipage}{0.8\textwidth}
\begin{enumerate}[label=(\Alph*), leftmargin=1cm]
  \item\label{cond:still-buy} $y \in D_i(\tilde p, b_i)$ and $y \ge c^{(i)}$, and
  \item \label{cond:higher-price} $p^{(i)}\le \tilde p\le
          (1+\epsilon)p$, and $\tilde p_j = (1+\epsilon)p_j$ whenever $y_j > \left(1+\epsilon\right) c^{(i)}_j $.
\end{enumerate}
\end{minipage}
\end{tabular}
\end{minipage}}\end{center}

\smallskip
The new individual prices will be set as $\tilde p$ and the new desired bundle as $y$.
Property \ref{cond:higher-price} requires that if 
 agent $i$ wants significantly more of good $j$ than the current amount $c^{(i)}_j$, then they are  willing to pay the higher price $(1+\epsilon)p_j$.

\medskip
The auction algorithm (Algorithm~\ref{exchangeMarkets}) considers the agents one-by-one in steps. 
A step gives an agent $i$ a chance to spend her surplus money $s_i$ to obtain more goods. 
If $s_i>0$, agent $i$ calls  \pFindNewPrices{($i, p^{(i)}, p, b_i, c^{(i)}$)} to obtain new individual prices $\tilde p$
and desired bundle $y$. At the end of their step, they update $p^{(i)}=\tilde p$. 

Given $\tilde p$ and $y$, agent $i$ considers  all goods with $\tilde p_j=(1+\epsilon)p_j$ one-by-one, and tries to purchase 
$y_j-c^{(i)}_j$ amount using the \pOutbid{} procedure.
First, if $\sum_{i=1}^n c^{(i)}_j<e_j$, i.e., if there is any unsold amount of good $j$, they purchase from such amounts.
If they still want more, they will outbid other agents who have been paying the
lower price $p_j$ for this good, 
by offering the higher price $(1+\epsilon)p_j$. 
Goods with  $\tilde p_j<(1+\epsilon)p_j$ are ignored: no additional amount is purchased from such goods.

If after the calls to \pOutbid{}, a good $j$ is only being sold at the higher price $(1+\epsilon)p_j$, then we call the \pRaisePrice{} procedure to increase the market price from $p_j$
to $(1+\epsilon)p_j$, and update the budgets and surpluses accordingly.
The algorithm terminates once the total surplus of the agents drops below $3 \epsilon  \pr{p}{e}$; according to Lemma~\ref{lem:term}, the 
current prices and allocations 
form a $4\epsilon$-approximate equilibrium.

\medskip
We now formulate the main running time statement. This depends on the running time $T_F$ of the subroutine \pFindNewPrices{}.
We assume that $T_F=\Omega(m)$, since the output needs to return an $m$-dimensional vector of goods.

We also use an upper bound on $p_{\max}/p_{\min}$---the ratio   of the largest and smallest nonzero prices at some $\epsilon$-equilibrium. 
An upper bound on $p_{\max}/p_{\min}$ may be obtained for the specific demand systems, e.g., 
for demand systems arising from linear utilities~\cite{GargK06}.
Alternatively, one can follow the approach of 
Codenotti, McCune, and Varadarajan~\cite{codenotti2005market,codenotti2005polynomial} by adding a dummy
agent with a Cobb--Douglas demand system and an initial endowment of a
small fraction of all goods. In the presence of such an agent, we can
obtain a strong bound on
$p_{\max}/p_{\min}$, at the expense of obtaining a slightly worse approximation
guarantee. We describe the construction in
Appendix~\ref{section:dummy}. 

Note that it is the ratio $p_{\max}/p_{\min}$ rather than the actual values of  $p_{\max}$ and $p_{\min}$ that matter: by Assumption~\ref{assumption1}, for (approximate-)equilibrium prices $p$, $\alpha p$ also gives (approximate-)equilibrium prices with the same allocation, 
for any $\alpha>0$. In our algorithm, the minimum price will remain $1$ throughout, see Lemma~\ref{lem:minprice}.
\begin{theorem}\label{thm:running-oracle}
Assume all agents have WGS demand systems.
Let $T_F$ be an upper bound on the running time of the
subroutine \textnormal{\pFindNewPrices{}}  and assume that $T_F = \Omega(m)$.
Algorithm~\ref{exchangeMarkets} finds a $4\epsilon$-approximate market
equilibrium in time 
$O\left(\frac{nmT_F}{\epsilon^2}\cdot \log
  \left(\frac{p_{\max}}{p_{\min}}\right)\right).$
\end{theorem}

One particular implementation of \pFindNewPrices{} in Section~\ref{section:newPrices-bounded} is given for bounded elasticities, see Lemma~\ref{lem:findNewPrices-running} for the running time bound.
Recall the elasticity bound $f$ from Definition~\ref{def:flexible}. 
\begin{theorem}\label{thm:alg-main}
If all agents have WGS demand systems with elasticity at least $-f$ for some $f>0$, then an $\epsilon$-approximate equilibrium can
be computed in time 
$O\left(\frac{nm^2 f\cdot T_D}{\epsilon^2}\cdot \log
  \left(\frac{p_{\max}}{p_{\min}}\right)\right),$
where $T_D$ is the time needed for one
call to the demand oracle. 
\end{theorem}

We give an overview of the running times of the previous auction
algorithms in Appendix~\ref{section:running-times}.

%%%%%%%%%%%%%%%%%% Main Algorithm %%%%%%%%%%%%%%%%%%%%
\begin{algorithm}[!htb]
\raggedright
\SetKwProg{Init}{Initialize}{}{}
\DontPrintSemicolon
\SetAlgoLined
\let\oldnl\nl
\newcommand{\nonl}{\renewcommand{\nl}{\let\nl\oldnl}}
    
\KwIn{ Demand systems $D_i$, and the endowment vectors $e^{(i)}$,
      and $\epsilon\in (0,0.25)$.} 
\KwOut{A $4\epsilon$-approximate market equilibrium.}
\Init{}{
\lFor{$j\in[m]$}{$e_j \gets \sum_{i=1}^n e^{(i)}_j$ ; $p_j \leftarrow 1$ ; $w_j\gets e_j$ ; $l_j\gets 0$} 
\For{$i\in[n]$}{$b_i\gets \sum_{j=1}^n e^{(i)}_j$ ; $s_i\gets b_i$ \;
\lFor{$j\in [m]$}{$p^{(i)}_j\gets 1$ ; $c^{(i)}_j\gets 0$ }}
}
    \While{$\sum_{i=1}^n s_i \le 3\epsilon \pr{p}{e} $}{  
        Select next agent $i\in [n]$ with $s_i > 0$. \tcp*[l]{Step for agent $i$.}\label{aStep}
        ($\tilde p, y) \leftarrow $  \pFindNewPrices($i,p, p^{(i)},  b_i, c^{(i)}$)\;
        \For{$j = 1$ \KwTo $m$}{
            \uIf (\tcp*[f]{Case 1}){  $p^{(i)}_j < (1+\epsilon )p_j$ and $ \tilde p_j = (1+ \epsilon)p_j$}{
            $s_i \leftarrow s_i - c^{(i)}_j \cdot\epsilon p_j$ ;
            $l_j \leftarrow l_j -c^{(i)}_j$  
            \tcp*[l]{$i$ pays $(1+\epsilon)p_j$ instead of $p_j$.}
            \pOutbid($i$, $j$, $y_j - c^{(i)}_j$)\;
            }\ElseIf (\tcp*[f]{Case 2}) {  $p^{(i)}_j = (1+\epsilon) p_j$ and $ \tilde p_j =  (1 + \epsilon)p_j$}{
            \pOutbid($i$, $j$, $y_j - c^{(i)}_j$)\;
            }
            \tcp*[h]{Skip the goods with $p^{(i)}_j < (1+\epsilon) p_j$ and $ \tilde p_j < (1+\epsilon) p_j$.}\tcp*[r]{Case 3}
            \lIf{$w_j + l_j = 0$}{\pRaisePrice($j$)}
        }
    $p^{(i)} \leftarrow \tilde p$\;
    }
    \Return $p$, $\{p^{(i)}\}_{i\in[n]}$ and $\{c^{(i)}\}_{i \in [n]}$\label{sum-s}
  \caption{Auction algorithm\label{exchangeMarkets}}
\end{algorithm}

 %%%%%%%%%%%%%%%%%%%%%%%%% Outbid %%%%%%%%%%%%%%%%
\begin{procedure}[h]
\DontPrintSemicolon
    \SetAlgoLined   
       \KwIn{ Agent $i\in[n]$, good $j\in [m]$, amount $\gamma>0$.} 
    $z\gets \gamma$\;
    \If(\tcp*[f]{a part of $j$ is unsold}){$w_j > 0$}{\label{if-first}
        $\mu \gets \min \{w_j, z\}$\;
        $w_j \leftarrow w_j - \mu$\;
         $c^{(i)}_j \leftarrow c^{(i)}_j
        + \mu$\;
      $s_i \leftarrow s_i -  (1+\epsilon)\mu $\tcp*[l]{here $p_j = 1$}\label{if-last}
       $z \leftarrow z- \mu$\;
    }
    \While {$z > 0$ and $l_j > 0$}{\label{l:while-first}
        Let $k\in [n]$ be such that  $c^{(k)}_j > 0$ and $j\in L_k$ \label{line:select-agent} \;
         $\mu\gets  \min \{c^{(k)}_j, z\}$\;
         $l_j \leftarrow l_j - \mu$ \;
          $c^{(k)}_j \leftarrow c^{(k)}_j - \mu$;\ $c^{(i)}_j \leftarrow c^{(i)}_j + \mu$\label{line:outbid-zero}\tcp*[l]{$i$ outbids $k$}
        $s_k \leftarrow s_k + \mu  p_j$;\ $s_i \leftarrow s_i - (1+\epsilon)\mu  p_j$\; 
        \label{l:while-last}
         $z \leftarrow z - \mu$\; 
    }
\caption{Outbid($i$, $j$, $\gamma$) \label{procOutbid}}
\end{procedure}

 %%%%%%%%%%%%%%%%%%%%%%%%% Outbid %%%%%%%%%%%%%%%%
\begin{procedure}[h]
\DontPrintSemicolon
    \SetAlgoLined   
       \KwIn{ Good $j\in [m]$.} 
       \For{$k\in [n]$}{
            $b_k \leftarrow b_k + \epsilon p_je^{(k)}_j$\label{budgetUpdate}\;
            $s_k \leftarrow s_k + \epsilon p_je^{(k)}_j$\;
            $p^{(k)}_j \leftarrow (1+\epsilon)p_j$\label{priceIncreaseNext}\;
            }
        $p_j \leftarrow (1+\epsilon)p_j$; \,$l_j\gets e_j$\,  \label{priceIncrease}      
\caption{RaisePrice($j$) \label{procRaiseprice}}
\end{procedure}

\subsection{Description of the algorithm}
We now give a more detailed overview of the algorithm. 
Recall the notation and invariants \ref{prop:price-inc}--\ref{prop:pay} at the start of this section. 

For each good $j=1,2,\ldots,m$, we partition the total amount as $e_j=w_j+l_j+h_j$ according to the price it is sold at:
\begin{itemize}
	\item amount $w_j$ is the unsold part of the good,
	\item amount $l_j$ is sold at the lower price $p_j$, and
	\item amount $h_j$ is sold at the higher price $(1+\epsilon)p_j$.
\end{itemize}
We only explicitly maintain $w_j$ and $l_j$ in the algorithm.
Recall from \ref{prop:fully-sold} that $w_j\ge 0$ and $p_j=1$ as long as $w_j>0$. 
We further maintain  $w_j + l_j > 0$ at the beginning of every step, i.e., there is always a part of the good
that is unsold or owned by an agent at the lower price. 

\paragraph{The Outbid subroutine}
Procedure~\pOutbid$(i,j,\gamma)$,
controls how the ownership of goods may change. 
This is called when agent $i$ would like to purchase an additional amount $\gamma$ of good $j$.

If $w_j>0$, there is some unsold amount of the good, 
then agent $i$ starts by purchasing $\min\{w_j,\gamma\}$ of this amount. 
Recall that $p_j=1$ at this point due to invariant \ref{prop:fully-sold}.

If agent $j$ still requires more of good $j$, 
we consider agents $k$ one-by-one who are paying the lower price $p_j$ for good $j$, i.e., $j\in L_k$.
Agent $i$ may take over some of this amount by offering a higher
price  $(1+\epsilon)p_j$.

\paragraph{The RaisePrice subroutine}
Procedure~\pRaisePrice$(j)$ is called when $w_j+l_j=0$ for a good $j$, i.e., it is only sold at the higher price $(1+\epsilon)p_j$. In this case, we increase the market price to $(1+\epsilon)p_j$, set all individual prices $p_j^{(k)}$ to this value, and 
update the budgets and surpluses of all agents whose initial endowment contains good $j$. We also set $l_j=e_j$.

\paragraph{Steps}
The algorithm terminates as soon as the
total surplus drops below $3\epsilon \pr{p}{e}$.
We consider agents $i$ with $s_i>0$ one-by-one. 
By invariant \ref{prop:dominate}, the agent is buying a bundle
$c^{(i)}\le x^{(i)}$ for some $x^{(i)} \in D_i(p^{(i)}, b_i)$. 
The subroutine \pFindNewPrices{($i, p^{(i)}, p, b_i, c^{(i)}$)}
delivers new prices $\tilde p$ and a bundle $y$ satisfying  
Conditions \ref{cond:still-buy} and \ref{cond:higher-price} described above.

Conditions \ref{cond:still-buy} asserts 
 that the current bundle $c^{(i)}$ of agent $i$ is still dominated by a desired bundle $y$ at the increased prices $\tilde p$.
Condition \ref{cond:higher-price} guarantees that $\tilde p \ge p^{(i)}$, and whenever an agent wants
to buy more of some good than they already own at least by a factor
$(1+\epsilon)$, then they are willing to pay the higher price $(1+\epsilon)p_j$ for it. 
(They might already be paying the increased price to start with if $p_j^{(i)} = (1+\epsilon)p_j$. 
In this case $\tilde p_j = (1+\epsilon)p_j = p_j^{(i)}$.)

The above properties suggest the following update rules for each good
$j\in [m]$.

\smallskip
\noindent \textit{Case 1. $p^{(i)}_j < (1+\epsilon) p_j$ and $\tilde p_j = (1+\epsilon)p_j$.} 
The good $j$ was in $L_i$ and needs to be moved to $H_i$, i.e., agent $i$ used to pay $p_j$ but now is willing to pay the higher price for $j$.  
We charge agent $i$ the price $(1+\epsilon)p_j$ for the amount $c^{(i)}_j$ they already own instead of $p_j$.
Additionally, agent $i$ outbids on good $j$ the amount $\min\{y_j-c^{(i)}_j,w_j+l_j\}$.

\smallskip
\noindent\textit{Case 2. $p^{(i)}_j = (1+\epsilon)p_j$ and $\tilde p_j = (1+\epsilon)p_j$.}
The good $j$ was in $H_i$ and stays in $H_i$, i.e., agent $i$ continues to pay the higher price.
The agent $i$ keeps the amount $c^{(i)}_j$ of good $j$ that they already had and
outbids for as much as they can from the other agents, i.e., the amount $\min\{y_j-c^{(i)}_j,w_j+l_j\}$.

\smallskip
\noindent\textit{Case 3. $p^{(i)}_j < (1+\epsilon)p_j$ and $\tilde p_j < (1+\epsilon) p_j$.} 
The good $j$ remains in $L_i$, i.e., agent $i$ continues to pay the lower price. 
By \ref{cond:higher-price}, we must have
$c^{(i)}_j \le y_j \le (1+\epsilon) c^{(i)}_j$; the agent will
not seek to buy more of these goods.

If the above updates result in $w_j + l_j = 0$ for good $j$, then we call \pRaisePrice$(j)$ to increase the market price.
Once all of the goods have been considered we set $p^{(i)}= \tilde p$ and update 
$c^{(i)} $ as the current allocation.

\paragraph{Rounds} 
In the analysis, it will be useful to organize the steps into rounds as in~\cite{GargK06}. 
A \emph{round} consists of all agents making a step once (or being skipped if $s_i = 0$). 

\subsection{Analysis}
We start with a high-level overview of the analysis.  Lemma~\ref{lem:invariants} verifies that all invariants~\ref{prop:price-inc}--\ref{prop:pay} are maintained. Lemma~\ref{lem:minprice} shows that the minimum price remains $p_j=1$ throughout. Note that---in accordance with invariant \ref{prop:fully-sold}---it suffices to show that $w_j>0$ for some good $j$ holds at any point. In other words, the algorithm terminates if all goods are fully sold.

Lemma~\ref{lemma:roundsInIteration} is the key argument in the running time bound: 
it shows that the number of rounds between two calls to \pRaisePrice{} is bounded by $2/\epsilon$.  
The idea is that as long as the market prices do not increase,  the total budget $(\sum_{i=1}^n b_i)$ remains the same, 
while the total money spent on the goods is increasing due to outbidding. 
In every outbid, an agent pays more by a factor $(1+\epsilon)$ for the goods they purchase. Thus, in $n$ consecutive steps the total surplus decreases approximately by a factor $(1+\epsilon)$.
This eventually leads to either a market price increase or reaches an approximate equilibrium.

\begin{lemma}\label{lem:invariants}
If all agents have WGS demand systems, then the invariants
\ref{prop:price-inc}-\ref{prop:pay}  hold after every step.
\end{lemma}
\begin{proof}
\begin{description}
\item[~\ref{prop:price-inc}] This is immediate.
\item[~\ref{prop:fully-sold}] The algorithm maintains  $e_j=w_j+l_j+h_j$ for all goods $j=1,2,\ldots,m$. Further, $w_j,l_j,h_j\ge 0$.  The amount sold is $l_j+h_j\le e_j$, hence, no good is oversold. 
The amount $w_j$ is monotone decreasing throughout. This is guaranteed by property \ref{cond:still-buy} of the procedure
\pFindNewPrices{}, and the fact that $c^{(i)}_j$ may only decrease if
another $c^{(k)}_j$ increases by the same amount.
 Thus, if $w_j=0$ at any point of the algorithm, good $j$ remains fully sold in all remaining steps. 
Note that the market price $p_j$ is first increased from the initial value $p_j=1$ once $w_j=l_j=0$.
\item[~\ref{prop:individual}] The budgets are updated in \pRaisePrice{}, adding $\epsilon p_j e^{(i)}_j$ to $b_i$ when the price $p_j$ increases to $(1+\epsilon)p_j$. Thus, we maintain $b_i = \pr{p}{e^{(i)}}$.
The bounds $p\le p^{(i)}\le (1+\epsilon)p$ are immediate from condition~\ref{cond:higher-price} 
in \pFindNewPrices{} and in the procedure \pRaisePrice{}.
\item[~\ref{prop:dominate}]
Suppose these properties hold for every agent before a step of agent $i$. 
The requirements~\ref{cond:still-buy} and \ref{cond:higher-price} guarantee that $c^{(i)}$ is dominated by a bundle $x^{(i)}\in D_i(p^{(i)}, b_i)$ and prices satisfy $p\le p^{(i)}\le (1+\epsilon)p$, for each agent $i$. 
Moreover, if the budget $b_i$ is increased in line~\ref{budgetUpdate}, the invariant remains true by the WGS property. 

Now, consider an agent $k$ different from $i$. 
In the step, $k$ could lose a part of a good only through the
outbid and hence $c^{(k)}$ does not increase. As long as the prices $p^{(k)}$
do not change, \ref{prop:dominate} holds trivially.
The only time $p^{(k)}$ can change is the price increase step in line~\ref{priceIncreaseNext}, namely, if
$p_j$ increases to $(1+\epsilon)p_j$, it forces
$p^{(k)}_j=(1+\epsilon)p_j$. Note that the price increase only happens once $l_j=0$.
Assume we had $p^{(k)}_j<(1+\epsilon)p_j$ before the price increase,
that is, agent $k$ was buying good $j$ at the lower price $p_j$. By
$l_j=0$ and invariant \ref{prop:pay}, it follows that $c^{(k)}_j=0$ at this point. 
As the budgets may only increase, the WGS property
implies that after increasing $p^{(k)}_j$, the bundle $c^{(k)}$ will
still be dominated by an optimal bundle.
\item[~\ref{prop:pay}] It is straightforward to check that the form of the surplus is maintained. \qedhere
\end{description}
\end{proof}

\begin{lemma}\label{lem:minprice}
We have $\min\{p_j:\, j\in G\}=1$ at the beginning of every step. 
\end{lemma}
\begin{proof}
We show that a good $j$ with $w_j > 0$ exists in every step. 
The statement then follows by  \ref{prop:fully-sold}.
For a contradiction, suppose we reach the point where $w_j=0$ for all $j=1,2,\ldots,m$, and consider the first time this happens. 
At this point all of the goods are fully sold, i.e., $\sum_{i =1}^n c^{(i)} = e$.
Consider the total surplus at this point. 
We have 
\begin{align*}
\sum_{i = 1}^n s_i &= \sum_{i=1}^n \left(b_i - \sum_{j \in L_i} p_j c^{(i)}_j - \sum_{j \in H_i} (1+\epsilon)p_j c^{(i)}_j\right) 
\le \sum_{i=1}^n \left(b_i -\pr{p}{c^{(i)}}\right)\\ 
&= \sum_{i=1}^n b_i - \pr{p}{\sum_{i=1}^n c^{(i)}} = \sum_{i=1}^n b_i - \pr{p}{e} = 0\,.
\end{align*}
This contradicts the assumption $\sum_{i=1}^n s_i > 3\epsilon \pr{p}{e} $ at the beginning of every step.
\end{proof}

%%%%%%%%%%%%%%%%%%%%%%%%%%%%%%%%%%%%%%%%%%%%%%%%
\begin{lemma}\label{lemma:roundsInIteration}
The number of rounds between two consecutive calls to \textnormal{\pRaisePrice{}} is at most $ 2/\epsilon$.
\end{lemma}

\begin{proof}
Let $p$ be the market prices after a call to \pRaisePrice{}, and consider a sequence of steps with the same market prices; 
consequently, the budget of every agent remains the same.
Consider a step of an agent $i$ during this sequence.
If $i$ buys $w_j+l_j$ of a good $j$, then \pRaisePrice{} is called and the sequence finishes. 
Thus, we can assume that during this sequence, every agent $i$ gets the amount of each good they desire.

Let $\varphi$ denote the total amount of money spent at a certain point
of this sequence of steps that is spent by the agents on higher price
goods. That is, 
\[
\varphi=(1+\epsilon)\sum_{i=1}^n \sum_{j\in H_i} c^{(i)}_j p_j.
\]
\begin{claim}\label{claim:surplus}
Let $s_i$ denote the surplus of agent $i$ at the beginning of their
step.
Then  the value of $\varphi$ increases at least by $(1+\epsilon)^2 s_i-2.25\epsilon b_i$ during agent $i$'s step.
\end{claim}
\begin{claimproof}
Recall Cases 1-3 in the description of the step. Let $T_k$ be the set of
goods that fall into case $k$, that is, $T_1\cup T_2\cup T_3=[m]$.
\begin{itemize}
\item  
If $j\in T_1$, then $(1+\epsilon)p_jy_j$ will be added to
$\varphi$ in the \pOutbid{} subroutine: In this case, the agent also
outbids itself, moving  good $j$ from $L_i$ to $H_i$.
\item If $j\in T_2$, then $(1+\epsilon)p_j( y_j-c^{(i)}_j)$
  will be added to $\varphi$ in the \pOutbid{} subroutine.
\item
If $j\in T_3$, then we do not increase
$\varphi$. Nevertheless, \ref{cond:higher-price} guarantees
that 
$\tilde  p_j (y_j- c^{(i)}_j)\le \epsilon \tilde p_j c^{(i)}_j.$
Consequently,
\begin{equation}\label{eq:T3}
\sum_{j\in T_3}\tilde  p_j (y_j- c^{(i)}_j)\le \epsilon \pr{\tilde p}{c^{(i)}}\, .
\end{equation}
\end{itemize}
Also note that $\tilde p_j=(1+\epsilon)p_j$ if $j\in T_1\cup T_2$.
Assumption~\ref{non-satiation} on non-satiation guarantees that
$\pr{\tilde p}{y}=b_i$. Let $\Delta\varphi$ denote  the increment in
$\varphi$; this can be lower bounded as
\[
\begin{aligned}
\Delta\varphi&=\sum_{j\in T_1} \tilde p_j y_j+\sum_{j\in T_2} \tilde
p_j(y_j-c^{(i)}_j) \pr{\tilde p}{y} -  \sum_{j\in T_3} \tilde p_j y_j - \sum_{j\in T_2} \tilde
p_jc^{(i)}_j \\
&\ge b_i - \sum_{j\in T_3} \tilde p_j (y_j - c^{(i)}_j) -\pr{\tilde
p}{c^{(i)}}
\ge b_i-(1+\epsilon) \pr{\tilde p}{c^{(i)}}\, ,
\end{aligned}
\] 
using \eqref{eq:T3}.
The money spent by the agent at the beginning of the step is
$b_i-s_i$. Good  $j$ is purchased at price
at least $p_j$ according to \ref{prop:pay}, and $\tilde p_j\le
(1+\epsilon) p_j$. 
Consequently, $\pr{\tilde p}{c^{(i)}}\le (1+\epsilon) (b_i-s_i)$. With
the above inequality, we obtain
\[
\Delta\varphi\ge b_i - (1+\epsilon)^2 (b_i-s_i)\ge (1+\epsilon)^2 s_i - (2\epsilon + \epsilon^2
)b_i \ge (1+\epsilon)^2 s_i - 2.25 \epsilon b_i,
\]
as $\epsilon< 0.25$. This completes the proof.
\end{claimproof}

As long as $\sum_{i=1}^n s_i>3\epsilon \pr{p}{e}$, the claim guarantees that $\varphi$ increases in every round by at least 
\[3(1+\epsilon)^2\epsilon \pr{p}{e} - 2.25\epsilon \sum_{i=1}^n b_i = 0.75 \epsilon \pr{p}{e}\,.\] 
Since
$\varphi\le (1+\varepsilon) \pr{p}{e}$,
the number of rounds until the next call to the \pRaisePrice{} is bounded
by $2/\epsilon$.
\end{proof}

We need one more lemma to bound the total running time of \pOutbid{}.
\begin{lemma}\label{lem:while-bound}
Between two calls to \textnormal{\pRaisePrice{}}, for every agent $k\in[n]$ and good $j\in [n]$, there can be only one occasion that during a call to \textnormal{\pOutbid{($i,j,\gamma$)}} for some $i\in[n]$ and $\gamma>0$, $c_j^{(k)}$ is set to 0 in line~\ref{line:outbid-zero}.
\end{lemma}
\begin{proof}
Assume  $c_j^{(k)}$ was set to 0 in a certain call to \pOutbid.
Before it can be set to 0 again,  agent $k$ must have obtained a positive amount of good $j$. This could only happen in a call to \pOutbid{($k,j,\gamma$)} for some $\gamma>0$. However, until the next call to \pRaisePrice{}, agent $k$ may only buy good $j$ at the higher price $(1+\epsilon)p_j$, and henceforth $j\in H_k$. Thus, agent $k$ cannot be selected again in line~\ref{line:select-agent} at a call to \pOutbid{($i,j,\gamma$)}.
\end{proof}

\begin{proof}[Proof of Theorem~\ref{thm:running-oracle}]
Lemma~\ref{lem:term} shows that at termination, the algorithm returns a $4\epsilon$-market equilibrium.

We first bound the total running time between two calls to \pRaisePrice{}. By Lemma~\ref{lemma:roundsInIteration}, there are at most $2/\epsilon$ rounds. Every round comprises $n$ steps, and every steps calls the procedure \pFindNewPrices{} exactly once.
Therefore, the time taken by \pFindNewPrices{} during this sequence of steps is $O(nT_F/\epsilon)$.

The total number of calls to \pOutbid{} is $m$ in each step, totaling to $O(nm/\epsilon)$. We bound the number of repeats in the `while' loop (lines~\ref{l:while-first}--\ref{l:while-last}) in all calls to \pOutbid{} between two calls to \pRaisePrice{}.
In a call to  \pOutbid{($i,j,\gamma$)}, in all but the final call to the `while' loop, we set $c^{(k)}_j=0$ for some agent $k$. By Lemma~\ref{lem:while-bound}, the total number of these events is at most $O(nm)$.
Hence, the number of repeats in the `while' loop between two calls to \pRaisePrice{} is  $O(nm/\varepsilon+nm)=O(nm/\varepsilon)$. Each repeat takes $O(1)$ time.

From the above, the total time of the \pOutbid{} calls is $O(nm/\epsilon)$ between two calls to \pRaisePrice{}. A call to  \pRaisePrice{} takes $O(nm)$ time.
Consequently, the
total time of such a sequence of steps is $O(nT_F /\epsilon+nm/\epsilon) = O(nT_F /\epsilon)$, using the assumption that $T_F = \Omega(m)$.

By Lemma~\ref{lem:minprice}, the minimum price remains at most $1$ throughout and therefore 
$p_{\max}$ is at most $p_{\max}/p_{\min}$.
Consequently,  \pRaisePrice{} can be called at most
$O(m \log_{1+\epsilon}(p_{\max}/p_{\min}))=O(\frac m\epsilon \log (p_{\max}/p_{\min}))$ times. The
claimed running time bound follows.
\end{proof}

%% file: FindingNewPrices.tex
%!TEX root = mainPlain.tex
%%%%%%%%%%%%%%%%%%%%%%%%%%%% Find New Prices %%%%%%%%%%%%%%%%%%%%%%%%
\section{Implementing the price update subroutine }\label{sec:findnewprices}
In this section, we present different approaches to implement \pFindNewPrices({$i, p, p^{(i)}, b_i, c^{(i)}$}).
Recall that the output prices $\tilde p\in\R^m_+$ and allocations $y\in\R^m_+$ must satisfy the following two requirements:
\begin{enumerate}[label=(\Alph*), leftmargin=1cm]
  \item $y \in D_i(\tilde p, b_i)$ and $y \ge c^{(i)}$, and
  \item $p^{(i)}\le \tilde p\le
          (1+\epsilon)p$, and $\tilde p_j = (1+\epsilon)p_j$ whenever $y_j > \left(1+\epsilon\right) c^{(i)}_j $.
\end{enumerate}

\smallskip
First, in Section~\ref{section:newPrices-bounded} we consider the setting of bounded elasticities. Recall from Lemma~\ref{lem:CES-flex} that this includes Cobb--Douglas utilities and CES utilities with parameter $\sigma>1$. Further, according to Lemma~\ref{lemma:convexComb}, we can use it for convex combinations of demand systems with bounded elasticities, even if they are not given in an explicit form \eqref{prog:optimalBundle}. The algorithm is a simple price increment procedure, making repeated calls to the demand oracle.
Whereas linear utilities do not have bounded elasticities, in
Section~\ref{sec:lin-price}, we give a simple direct algorithm for linear demand systems. 
Finally, in Section~\ref{section:gale}, we implement \pFindNewPrices{} for Gale demand systems. We obtain the new prices as the optimal Lagrangian multipliers of a convex program.

Note that for many utility functions, such as Cobb--Douglas or CES, we can use either of the methods in Section~\ref{section:newPrices-bounded} or  Section~\ref{section:gale}. The running time in Section~\ref{section:newPrices-bounded} depends linearly on the elasticity parameter $f$ and makes several calls to the demand oracle. Still, it could be faster than solving a convex program, e.g., if the demand oracle is given by an explicit formula.

It is possible to find further direct approaches  for particular demand systems, similarly to the approach in Section~\ref{sec:lin-price} for linear demand systems. For example, it is easy to devise an $O(m)$ time procedure for 
Cobb--Douglas demand systems, exploiting the fact that the optimal bundle allocates $\alpha^{(i)}_j b_i$ money for good $j$. Hence, each
price can be set independently of the others. 

We note that even though the above cases cover all standard examples of WGS systems, we do not have a general implementation for demand systems in the form \eqref{prog:optimalBundle}. 

\subsection{Demand systems with bounded elasticities}
\label{section:newPrices-bounded}
%%%%%%%%%%%%%%%%%%%%%%%%%%%%%%%%%%%%%%%%%%%%%%%%%%%%%%%%%%%%
Let us assume that the demand system $D_i$ has elasticity at least
$-f$ for some $f>0$. The subroutine \pBEFindNewPrices$_f$($i,p, p^{(i)},  b_i, c^{(i)})$ (Algorithm~\ref{algo:newPrices}) is a simple price increment procedure.
First, we obtain $y\in D_i(p^{(i)},b_i)$ from the demand oracle with $y\ge c^{(i)}$. 
This is possible due to invariant
\ref{prop:dominate}, which guarantees that $c^{(i)}\le x^{(i)}$ for
some $x^{(i)}\le D_i(p^{(i)},b_i)$. Thus,  $y=x^{(i)}$ is itself a suitable choice.
Then, we iterate the following step. As long as
\ref{cond:higher-price} is violated for a good $j$, we increase
its price by a factor $(1+\epsilon)^{1/f}$ until it reaches the upper
bound  $(1+\epsilon)p_j$.

\begin{algorithm}[t]
    \SetAlgoLined
    \DontPrintSemicolon
    \KwIn{Agent $i\in [n]$,  market prices $p\in\R^m_+$, individual prices $p^{(i)}\in\R^m_+$ such that $p\le p^{(i)}\le (1+\epsilon)p$, budget $b_i\in\R_+$, and bundle $c^{(i)}\in\R_+^m$.} 
    \KwOut{Prices $\tilde p$ and bundle $y$.}
    Initialization: $\tilde p \gets p^{(i)}$ \; 
  Obtain $y\in D_i(\tilde p, b_i)$ from the demand oracle such that
   $y\ge c^{(i)}$ \;
    \While {$\exists j:  \tilde p_j < (1+\epsilon)p_j$ and $y_j>(1+\epsilon)c^{(i)}_j$}{
       $\tilde p_j\gets \min \{ (1+\epsilon)^{1/f}\tilde p_j, (1+\epsilon)p_j\}$ \;
 Obtain  $y'\in D_i(\tilde p, b_i)$ from the demand oracle such that $y'_k\ge y_k$ for
       $k\neq j$ \;
  $y\gets y'$ \;
    }
\Return $(\tilde p, y)$ \;
\caption{\pBEFindNewPrices$_f$($i,p, p^{(i)},  b_i, c^{(i)}$)}\label{algo:newPrices}
\end{algorithm}

\begin{lemma}\label{lem:findNewPrices-running}
Assume the demand system $D_i$ has elasticity at least
$-f$ for some $f>0$.
Algorithm~\ref{algo:newPrices} terminates with $\tilde p$ and $y$
satisfying \ref{cond:still-buy} and \ref{cond:higher-price} in
time $O(mf\cdot T_D)$,
where $T_D$ is the time for a call to the demand oracle.
\end{lemma}
We assume that $T_D=\Omega(m)$, since the demand oracle needs to
output an $m$-dimensional vector.
\begin{proof}
The bound on the number of iterations is clear: since we have $p\le
\tilde p\le (1+\epsilon)p$ throughout, the price of every
good can increase at most $f$ times. Condition \ref{cond:still-buy} is
satisfied due to the WGS property and the bound on the demand
elasticity. When increasing $\tilde p_j$, the demand $y_k$ for $k\neq
j$ is non-decreasing as guaranteed by the demand oracle. Further,
$y_j$ may decrease only by a factor $(1+\epsilon)$, and since we had
$y_j>(1+\epsilon)c_j^{(i)}$ before the price update, we still have
$y_j>c_j^{(i)}$ after the price update. Condition \ref{cond:higher-price}
is satisfied at termination since the while loop keeps running as
long as
it is violated. Checking the while condition each time requires
$O(m)$ time; however, this will be dominated by the time $T_D$ according to
the comment on $T_D=\Omega(m)$ above. 
\end{proof}

\subsection{Linear demand systems}\label{sec:lin-price}
We now give a simple direct
implementation of \pFindNewPrices{} for linear demand systems.

\begin{lemma}\label{lemma:linearPrices}
\textnormal{\pFindNewPrices{}} can be implemented in $O(m)$ for a linear demand system
corresponding to the utility function $u(x) = \pr{v}{x}$.
\end{lemma}
\begin{proof}
Recall that for linear utilities $y\in D_i(\tilde p,b)$, $y_{j}>0$ if
and only if $j\in\arg\max_k v_k/p_k$, called \emph{maximum bang-per-buck
goods (MBB)}. We initialize  $\tilde p=p^{(i)}$, and
 let $S\subseteq [m]$ denote the set of
MBB goods. Thus, $y_j=0$ for all $j\notin S$.
We start increasing the prices of all
goods $j\in S$ at the same rate $\alpha$. Once a good outside $S$
becomes MBB, we include it in the set $S$ and also start raising its price.
We terminate when the budget is exhausted or when the price $\tilde p_k$ for a good $k\in S$ reaches 
the upper bound $(1+\epsilon)p_k$. In the latter case, we return the bundle
$y_j=c^{(i)}_j$ if $j\neq k$, and set $y_k= (b_i-\sum_{j\neq k}\tilde p_j
c_j)/\tilde p_k$; clearly, $y_k\ge c^{(i)}_k$. %\jnote{$y_k \ge c^{(i)}_k$?}
These prices and allocations
satisfy \ref{cond:still-buy} and \ref{cond:higher-price}; in fact,
we obtain \ref{cond:higher-price} in the stronger form that $\tilde
p_j=(1+\epsilon)p_j$ whenever $y_j>c^{(i)}_j$. We need to add a good to  $S$
at most $m$ times, and thus we can implement the procedure in $O(m)$ time.
\end{proof}

%%% Local Variables:
%%% mode: latex
%%% TeX-master: "main"
%%% End:

%% file: galeEquilibrium.tex
%!TEX root = mainFileMOR.tex

\subsection{Gale demand systems}\label{section:gale}
We now show that the subroutine \pFindNewPrices{} can be implemented
for Gale demand systems via convex programming. According to
Lemma~\ref{lem:gale-equiv}, 
this result is also applicable for demand systems given in the form 
\eqref{prog:optimalBundle} for
utility functions that are
homogeneous of degree one, in which case the optimal solutions to
\eqref{prog:optimalBundle}  and \eqref{galeObjective} coincide.

Let the utility function $u:\R^m_+\to \R_+$ be strictly concave and differentiable. Strict concavity implies that the demand system is simple: $|G^u(p,b)|=1$ for all $(p,b)\in \R^m_+$.

We implement a stronger and more general form of  \pFindNewPrices{}, with an arbitrary vector $q\in\R^m_+$, $q\ge p$ in place of $(1+\epsilon)p$.

Assume we are given $b\in\R_+$, $p,q,c \in \R^m_+$ such that $p \le q$, and moreover assume that $c\le x$ for some  $x \in G^u(p, b)$. The goal is to find $\tilde p$ and $y$ such that
\begin{enumerate}[label=(\Alph*\textsc{\char13})]
\item\label{firstg} $y \ge c$ where $y \in G^u(\tilde p, b)$, and
\item\label{secondg} $p\le \tilde p \le q$ and $\tilde p _j = q_j$ whenever $y_j > c_j$.
\end{enumerate} 

\smallskip
In the following convex program, the agent is allowed to buy a good $j$ at two prices: 
amount $y'_j$ at price $p_j$ and amount $y''_j$ at price $q_j$; 
the amount at the lower price $p_j$ is capped at $c_j$.
\begin{equation}\label{prog:GaleConstrained}
\begin{aligned}
\max~~  b \ln u(y) & - \pr{p}{y'} - \pr{q}{y''} \\
y&=y'+y''\\
y' &\le c \\
y', y'' &\ge  0 \, .
\end{aligned}
\end{equation}

We show that the optimal solution to this program, along with the
prices obtained from the KKT conditions satisfy  \ref{firstg} and \ref{secondg}.
\iffalse
\begin{lemma}
Let $u:\R^m_+\to \R_+$  be a monotone, strictly concave and
differentiable utility function. Let $y^* = y' + y''$ be the unique
optimal solution of~\eqref{prog:GaleConstrained}. Then, $y^*$
satisfies \eqref{firstg} and \eqref{secondg}.
\end{lemma}
\begin{proof}
\fi

Since all constraints are linear, strong duality holds.
Let $y^* = y' + y''$ be an optimal solution of~\eqref{prog:GaleConstrained}. By the KKT conditions, 
there exists $\alpha\in\R_+^m$ such that for any $j\in [m]$,
\begin{enumerate}[label=(\roman*)]
\item\label{l1} $b\cdot \frac{\partial_j u(y^*)}{u(y^*)} \le \min \{ \alpha_j + p_j, q_j\}$,
\item\label{l2} $b\cdot \frac{\partial_j u(y^*)}{u(y^*)} = \alpha_j + p_j$ whenever $y'_j>0$,
\item\label{l3} $b\cdot \frac{\partial_j u(y^*)}{u(y^*)} = q_j$ whenever
  $y''_j>0$, and 
\item \label{l4} $y'_j=c_j$ whenever $\alpha_j>0$.
\end{enumerate}
Note that in an optimal solution we must have $y'_j>0$ whenever $c_j>0$ and $y^*_j>0$.
We define the prices $\tilde{p}_j$ as 
\[
\tilde p_j :=
\left\{
  \begin{array}{ll}
    q_j & \mbox{if } c_j = 0 \mbox{ and } y^*_j > 0 \,, \\ 
    {\alpha_j}+p_j  & \mbox{otherwise}  \,.
  \end{array}
\right.
\] 

\begin{lemma} \label{lemma:GaleShadowPrices}
The allocations $y^*$ and prices $\tilde p$ satisfy \ref{firstg} and \ref{secondg}.
\end{lemma}
\begin{proof}
Since all constraints are linear, strong duality holds for \eqref{galeObjective} as well as for \eqref{prog:GaleConstrained}.
We start with \ref{secondg}. 
Let $j \in [m]$. If $c_j = 0$ and $y^*_j > 0$, then $\tilde p_j = q_j$ and thus \ref{secondg} holds by definition.

Suppose $c_j > 0$. In case $y''_j>0$, \ref{l2} and \ref{l3} imply that $\tilde p_j=\alpha_j+p_j=q_j$. This holds whenever $y^*_j>c_j$. 
It is left to show that that $p_j\le \tilde p_j \le q_j$ in case  $y^*_j \le c_j$.
If $y^*_j > 0$, this follows from \ref{l1} and \ref{l2}. 
If $y^*_j = y'_j = 0$, by $c_j > 0$ and \ref{l4} we have $\alpha_j = 0$ and thus $\tilde p_j = p_j$.

\smallskip
For \ref{firstg}, we first show $y^*\in G^u(\tilde p,b)$. 
By the KKT conditions for \eqref{galeObjective}, we have $y^*\in G^u(\tilde p,b)$ if and only if for all $j \in [m]$ it holds: 
\begin{enumerate}[label=(G\arabic*)]
  \item\label{c:g1} $\frac{b \partial_j u(y^*)}{u(y^*)} \le \tilde p_j$, and 
  \item\label{c:g2} $\frac{b \partial_j u(y^*)}{u(y^*)} = \tilde p_j$ whenever $y^*_j>0$.
\end{enumerate}
Let $j \in [m]$. 
The condition $\frac{b \partial_j u(y^*)}{u(y^*)} \le \tilde p_j$ follows by definition of $\tilde p$ and by \ref{l1}.
Suppose $y^*_j = y'_j + y''_j > 0$.
If $c_j=0$, then $y''_j>0$ and $\tilde p_j=q_j$; \ref{c:g2} follows by \ref{l3}. If $c_j>0$, then we must have $y'_j>0$ and $\tilde p_j=p_j+\alpha_j$. Thus, \ref{c:g2} follows from \ref{l2}.

It remains to show that $y^*\ge c$. We prove by contradiction: assume
that $y^*_j<c_j$ for some good $j$. In particular, $c_j > 0$.
This implies $\alpha_j=0$ by \ref{l4}, yielding $\tilde p_j=p_j$. 
By the strict concavity assumption, $y^*$ is the unique optimal bundle in $G^u(\tilde p, b)$. 
Using the WGS property for $(p,b)$ and $(\tilde p, b)$ we have
$y^*_j\ge x_j$ since $p_j = \tilde p_j$. We obtain a contradiction
to $y^*_j<c_j\le x_j$. 
\end{proof} 
%%% Local Variables:
%%% mode: latex
%%% TeX-master: "main"
%%% End:

%% file: NSW.tex
%!TEX root = mainPlain.tex
\section{Nash social welfare and spending restricted equilibrium}
\label{sec:fisher}

%\enote{This section also changed substantially}

In the rest of the paper we focus on the Nash social welfare (NSW) problem and related spending restricted (SR) equilibria. 
In this section, we provide the necessary definitions. For a more detailed introduction and significance of the NSW problem we refer to~\cite{anari2017nash,cole2015approximating,garg2018approximating}.

In the Nash social welfare (NSW) problem, the objective is to allocate $m$ indivisible goods to $n$ agents ($m\ge n$). Each agent is equipped with a utility function over the subsets of goods. The goal is to find a partition $S_1\cup S_2\cup \ldots S_n=[m]$ of the goods in order to maximize the geometric mean of the agents' utilities, $\left(\prod_{i=1}^nu_i(S_i)\right)^{1/n}$.

The NSW problem is NP-hard already for linear (additive) utilities, that is, if $u_i(S)=\sum_{j\in S} v_{ij}$ for each agent $i\in [n]$.  Here, we focus on the constant factor approximability of the problem.

The first constant factor approximation for this problem was given by Cole and Gkatzelis \cite{cole2015approximating} for the problem under linear utilities.  Their approach solves a continuous relaxation that corresponds to a divisible market problem, and rounds an optimal fractional  solution.  Recall that a continuous utility function for agent $i$ is give by $u_i(x)=\sum_{j\in S} v_{ij}x_{ij}$.  Then, the natural relaxation is exactly the program \eqref{eq:EG} with all $b_i=1$ and $e_j = 1$.  In other words, the natural relaxation is a market equilibrium in the induced Fisher market.  However, it is easy to see that this relaxation has an unbounded integrality gap.  To avoid this issue, Cole and Gkatzelis \cite{cole2015approximating} introduced the notion of {\em spending restricted equilibrium} that we now define in a slightly more general form. 

\begin{definition}[SR-equilibrium]\label{def:SR}
Consider an Fisher market with a set $A=[n]$ of agents, budgets $b_i$ for $i\in A$, a set $G=[m]$ of goods.  Let $D_i(p,b)$ denote the demand system of agent $i\in A$.  The prices $p\in \R^m_+$ and bundles $x^{(i)} \in \R^m_+$ form a \emph{spending restricted (SR) market equilibrium} if
 \begin{enumerate}[label=(\roman*)]
 \item  $x^{(i)}\in D_i(p,b_i)$ for all $i\in A$, and
\item  $\sum_{i=1}^n x^{(i)}_j\le a_j := e_j \cdot \min\{1, 1/p_j\}$, with equality whenever $p_j>0$, for all $j\in G$.
\end{enumerate}
For given prices $p$, we say that $a_j$ is the \emph{available amount} of good $j$.
\end{definition}

For each good $j$, the spending on a good $j$ in an SR-equilibrium is bounded by $e_j$.
We note that unlike the standard market equilibrium,
the set of SR-equilibria can be non-convex already for {\em capped
linear utilities} as shown in~\cite{garg2018approximating}. 
Capped linear utilities are defined as $u(x)=\min(c_i,\sum_j v_{j}x_j)$ for $(c, v)\in \R_+^{m+1}$.

Cole and Gkatzelis~\cite{cole2015approximating} first compute an
SR-equilibrium for linear utilities, 
and show that this can be rounded to an integer solution of
cost at most $2e^{1/e}$ times the optimal NSW solution. 
Spending restrictions cannot be directly added to~\eqref{eq:EG} since they involve the Lagrange multipliers $p$. 
In~\cite{cole2015approximating}, an SR-equilibrium was found via an extension of algorithms by Devanur et al.~\cite{DevanurPSV08} 
and Orlin~\cite{orlin2010improved} for market equilibria for Fisher markets with linear utilities.

Subsequent work by Cole
et al.~\cite{cole2017convex} showed that a spending restricted
equilibrium for the linear markets can be obtained as an optimal solution of a 
convex program (extending a convex formulation of linear Fisher market
equilibria that is different from~\eqref{eq:EG}), and
also improved the approximation guarantee to 2 (the current best factor 
is 1.45~\cite{barman2018finding}). However, this convex formulation is only known to
work for linear utility functions.

 Further work has studied the NSW problem for more general utilities, following the same strategy of first solving a
 SR-equilibrium problem then rounding. 
 Anari et al.~\cite{anari2018nash} studied NSW with 
 {\em separable, piecewise-linear concave (SPLC)} utilities. 
 Garg et al.~\cite{garg2018approximating} studied capped linear utilities. 
 Both papers find (exact or approximate) the corresponding SR-equilibria via fairly complex combinatorial algorithms. 

\paragraph{The Gale demand systems and NSW} The demand systems used in~\cite{anari2018nash,garg2018approximating} do not correspond to the demand system obtained by maximizing the utility function subject to the budget constraint~\eqref{prog:optimalBundle}. 
For capped linear utilities~\cite{garg2018approximating}, 
one needs additional condition that if an agent reaches their utility cap, 
then they will aim to minimize the money spent on achieving this level of utility. 
For SPLC utilities~\cite{anari2018nash}: intuitively,  each agent has a \emph{utility price} for each good and 
the agent is maximizing their utility subject to the budget constraint 
with respect to the utility prices (the utility price is always at least the market price).
In both cases, the total spending of the agents can be below their budgets. 
A natural unified way of capturing these equilibrium concepts is as follows. 
We require that at given prices $p$, agents maximize $\log u_i(x) - p^\top x $ instead of~\eqref{prog:optimalBundle}.
That is, the right notion for the NSW problem seems to be Gale demand system, which we restate here. 
\begin{equation*}
G^u (p, b) = \argmax_{x \in \R^m_+} \, b \log u(x) - p^\top x  \, .
\end{equation*}

Thus, the algorithms in~\cite{anari2018nash,garg2018approximating} 
compute and round an (approximate) SR-equilibrium under Gale demand system to an approximate solution of the NSW problem.
Interestingly, Gale demand systems arising from SPLC and capped additive utilities satisfy the WGS property 
as we will see in Section~\ref{section:BASPLC}.
In contrast, the demand systems arising from SPLC and capped additive utilities in the usual setting~\eqref{prog:optimalBundle} do not satisfy the WGS property.

\medskip
In the rest of the paper, we first modify the auction algorithm for SR-equilibria in Section~\ref{section:SR} 
and then in Section~\ref{section:BASPLC} we give the approximation algorithm for the NSW problem.
We will use the modified auction algorithm to find an approximate SR-equilibrium under Gale demand system of {\em capped SPLC} utilities, the common generalization of the utilities appearing in~\cite{anari2018nash} and~\cite{garg2018approximating}.
In order to apply the auction algorithm, we show that the Gale demand system of {\em capped SPLC} utilities satisfy the WGS property and show how to implement \pFindNewPrices{} in this case. 

Combining the above with a similar rounding as in~\cite{garg2018approximating}, we obtain a constant-factor approximation algorithm for maximizing NSW in polynomial-time when agents have capped SPLC utilities and goods come in multiple copies. The previous algorithm for this setting~\cite{ChaudhuryCGGHM18} runs in pseudopolynomial time. For the special case of linear utilities,~\cite{BeiGHM19} gives such an algorithm.

\paragraph{Existence of SR-equilibria} 
We note that whereas an equilibrium will always exist for WGS
utilities, the existence of an SR-equilibrium is a nontrivial
question. 
For example, suppose an agent $i$ has budget $b_i$ and 
Cobb-Douglas utility function $\prod_{j=1}^m (x^{(i)}_j)^{\beta_j}$,
where $\sum_j \beta_j = 1$, such that $\beta_j >\frac{e_j}{b_i}$ for
some $j$. 
Then the agent $i$ would like to spend at least $\beta_j b_i > e_j$ on
good $j$ for any prices $p$, but the total money that can be spent on
this good is $e_j$. 
Hence, there doesn't exist any SR-equilibrium in this case.

While we do not have general necessary and sufficient conditions on
the existence of an SR-equilibrium, we show that the models
previously studied in the context of NSW admit an {SR-}equilibrium.

\paragraph{Approximate spending-restricted equilibrium}
We use an extension of Definition~\ref{def:approx-eq}.
The main difference is that we
require  all goods to be fully consumed.

\begin{definition}[Approximate SR-equilibrium]\label{def:approx-SReq} 
Consider the same setting as in Definition~\ref{def:SR}.
For an $\epsilon>0$, prices $p\in \R^m$ and 
bundles $x^{(i)} \in \R^m_+$  form an \emph{$\epsilon$-approximate SR-equilibrium} if
\begin{enumerate}[label=(\Roman*)]
\item \label{def:SRsubset} $x^{(i)}\le z^{(i)}$ for some $z^{(i)}\in
  D_i(p^{(i)}, b_i)$, where $p\le p^{(i)}\le (1+\epsilon)
  p$, 
\item \label{def:SRno-over} $\sum_{i=1}^n x^{(i)}_j = a_j := e_j \cdot \min\{1, 1/p_j\}$ for all $j$, and
\item \label{def:SRleftover} $\sum_{j=1}^m p_j \left(\sum_{i=1}^n z^{(i)}_j - a_j\right)\le
  \epsilon \sum_{i=1}^{n} b_i$.
\end{enumerate}
\end{definition}
%%% Local Variables:
%%% mode: latex
%%% TeX-master: "main"
%%% End:

%% file: auctionAlgorithmSR.tex
%!TEX root = mainPlain.tex
\section{The auction algorithm for spending restricted equilibria}\label{section:SR}\label{section:spendingRestricted}
%\enote{This section is changed quite a bit.}
We present a modification of Algorithm~\ref{exchangeMarkets} for finding an approximate SR-equilibrium in a Fisher market where each agent satisfies the WGS property. The necessary changes are fourfold.
\begin{itemize}
\item The budgets $b_i$ are constant throughout the algorithm and are part of the input. As such, they do not depend on the prices of goods in the market. 
\item We need to make sure that in an SR-equilibrium exactly $e_j\cdot \min \{1, 1/p_j\}$ of a good is sold. 
\item The initialization needs to be changed since the prices are not scale-invariant as in exchange markets: We cannot assume that there exists an SR-equilibrium with $p_j \ge 1$ for all $j$.
\item  We do not make Assumption~\ref{non-satiation}
on non-satiation. We only use the following weaker assumption, namely
that after the prices increase, the spending of every agent is
non-decreasing. For capped linear utilities,
Assumption~\ref{non-satiation} does not hold, whereas this
weaker assumption is true.

\end{itemize}
\begin{assumption}\label{assumption:budgetIncreases}
Let $(p, b) \in \R^{m+1}$ and $x \in D(p, b)$. If $q \ge p$ and $y \in D(q, b)$, then $\pr{q}{y} \ge\pr{p}{x}$. 
\end{assumption}

Algorithm~\ref{SRauction} describes the algorithm for finding SR-equilibria when agents have WGS demands and uses
the same subroutine \pFindNewPrices{}.
It outputs a $4\epsilon$-approximate market equilibrium. 
We start by listing the invariants.

\begin{enumerate}[label=(SR\alph*)]
\item\label{SRprop:price-inc}
We maintain a price vector $p$ called  {\em market prices}.
Prices may only increase, and remain integer powers of  $(1+\epsilon)$.
\item\label{SRprop:fully-sold} The amount of each good $j$ being sold is exactly $a_j = e_j \cdot \min\{1, 1/p_j \}$.
\item\label{SRprop:individual} Every agent $i\in [n]$ maintains \emph{individual prices} $p^{(i)}\in\R_+^m$
  such that $p\le p^{(i)}\le (1+\epsilon)p$. We let
  $
  L_i  := \{j \in [m] : p^{(i)}_j < (1+\epsilon)p_j \}$ and $ H_i := [m] \setminus L_i\, .$
\item\label{SRprop:dominate} Every agent $i\in[n]$
   owns a bundle of goods
  $c^{(i)}\in\R_+^m$ that is dominated by a \emph{desired bundle} $x^{(i)} \in D_i(p^{(i)},b_i)$.
  \item\label{SRprop:pay} For the amount $c^{(i)}_j$ of good $j$, agent $i$ pays $p_j$ if $j\in L_i$ and $(1+\epsilon)p_j$ if $j\in H_i$.
  The \emph{relative surplus} of agent $i$ is $s^r_i := \pr{p^{(i)}}{x^{(i)}} -\sum_{j \in L_i} c^{(i)}_j  p_j - (1+\epsilon)\sum_{j\in
      H_i}   c^{(i)}_j  p_j$. 
\end{enumerate}
The relative surplus is the difference between the money the agent would like to spend
and what they are actually spending.
Under Assumption~\ref{non-satiation}, $s^r_i=s_i$ holds; we need to make
the distinction since we do not assume non-satiation.

\begin{lemma}\label{lem:SRterm}
Assume that \ref{SRprop:price-inc}--\ref{SRprop:pay} hold. Then $s^r_i\ge 0$ for all $i\in [n]$. Moreover, if 
$\sum_{i=1}^n s_i\le 3\epsilon (\sum_{i=1}^n b_i)$
then the prices $p$ and allocations $c^{(i)}$, $i\in[n]$ form a $4\epsilon$-approximate SR-equilibrium.
\end{lemma}
\begin{proof} 
By \ref{SRprop:individual}--\ref{SRprop:pay} it follows that $s^r_i \ge 0$.

Condition \ref{def:SRsubset} in Definition~\ref{def:approx-SReq} is immediate from invariants \ref{SRprop:individual} and \ref{SRprop:dominate}, and condition~\ref{def:SRno-over} follows from \ref{SRprop:fully-sold}. It is left to verify condition \ref{def:SRleftover}. We can write
\[
\begin{aligned}
 \pr{p}{\left(\sum_{i=1}^n x^{(i)} \right) - a}&=\sum_{i=1}^n \pr{p}{x^{(i)}}-\sum_{i=1}^n \pr{p}{c^{(i)}} \le\sum_{i=1}^n \pr{p^{(i)}}{x^{(i)}}-\sum_{i=1}^n \pr{p}{c^{(i)}}\\
&\le\sum_{i=1}^n\left( s_i +\epsilon \sum_{j\in H_i}   c^{(i)}_j  p_j\right)=\sum_{i=1}^n  s_i +\epsilon  \sum_{i=1}^n b_i \le 4\epsilon  \sum_{i=1}^n b_i\, . 
\end{aligned}\qedhere\]
\end{proof}
%%%%%%%%%%%%%%%%%% Main Algorithm %%%%%%%%%%%%%%%%%%%%
\begin{algorithm}[!htb]
\raggedright
\SetKwProg{Init}{Initialize}{}{}
\DontPrintSemicolon
\SetAlgoLined
\let\oldnl\nl
\newcommand{\nonl}{\renewcommand{\nl}{\let\nl\oldnl}}
    
\KwIn{Demand systems $D_i$, budgets $b_i$ and $\epsilon\in (0,0.25)$.} 
\KwOut{A $4\epsilon$-approximate SR-equilibrium.}
\Init{}{}
    \While{$\sum_{i=1}^n s^r_i \le 3\epsilon \pr{p}{e} $}{  
        Select next agent $i\in [n]$ with $s^r_i > 0$. \tcp*[l]{Step for agent $i$.}\label{SRaStep}
        ($\tilde p, y) \leftarrow $  \pFindNewPrices($i,p, p^{(i)},  b_i, c^{(i)}$)\;
        {$\tilde s := s^r_i + \pr{\tilde p}{y} - \pr{p^{(i)}}{x^{(i)}}$}\;
        \For{$j = 1$ \KwTo $m$}{
            \uIf (\tcp*[f]{Case 1}){  $p^{(i)}_j < (1+\epsilon )p_j$ and $ \tilde p_j = (1+ \epsilon)p_j$}{
            $\tilde s \leftarrow \tilde s - c^{(i)}_j \cdot\epsilon p_j$ ;
            $l_j \leftarrow l_j -c^{(i)}_j$  
            \tcp*[l]{$i$ pays $(1+\epsilon)p_j$ instead of $p_j$.}
            \pOutbid($i$, $j$, $y_j - c^{(i)}_j$)\;
            }\ElseIf (\tcp*[f]{Case 2}) {  $p^{(i)}_j = (1+\epsilon) p_j$ and $ \tilde p_j =  (1 + \epsilon)p_j$}{
            \pOutbid($i$, $j$, $y_j - c^{(i)}_j$)\;
            }
            \tcp*[h]{Skip the goods with $p^{(i)}_j < (1+\epsilon) p_j$ and $ \tilde p_j < (1+\epsilon) p_j$.}\tcp*[r]{Case 3}
            \lIf{$l_j = 0$}{\pRaisePrice($j$)}
        }
    $p^{(i)} \leftarrow \tilde p$ and update $s^r_i$\;
    }
    \Return $p$, $\{p^{(i)}\}_{i\in[n]}$ and $\{c^{(i)}\}_{i \in [n]}$\label{SRsum-s}
  \caption{Auction algorithm for SR-equilibrium\label{SRauction}}
\end{algorithm}

 %%%%%%%%%%%%%%%%%%%%%%%%% Outbid %%%%%%%%%%%%%%%%
\begin{procedure}[h]
\DontPrintSemicolon
    \SetAlgoLined   
       \KwIn{ Agent $i\in[n]$, good $j\in [m]$, amount $\gamma>0$.} 
    $z\gets \gamma$\;
    \While {$z > 0$ and $l_j > 0$}{\label{SRl:while-first}
        Let $k\in [n]$ be such that  $c^{(k)}_j > 0$ and $j\in L_k$ \label{SRline:select-agent} \;
         $\mu\gets  \min \{c^{(k)}_j, z\}$\;
         $l_j \leftarrow l_j - \mu$ \;
          $c^{(k)}_j \leftarrow c^{(k)}_j - \mu$;\ $c^{(i)}_j \leftarrow c^{(i)}_j + \mu$\label{SRline:outbid-zero}\tcp*[l]{$i$ outbids $k$}
          $s^r_k \leftarrow s^r_k + \mu  p_j$;\ $\tilde s \leftarrow \tilde s - (1+\epsilon)\mu  p_j$\; 
        \label{SRl:while-last}
         $z \leftarrow z - \mu$\; 
    }
\caption{Outbid($i$, $j$, $\gamma$) \label{SRprocOutbid}}
\end{procedure}

 %%%%%%%%%%%%%%%%%%%%%%%%% Outbid %%%%%%%%%%%%%%%%
\begin{procedure}[h]
\DontPrintSemicolon
    \SetAlgoLined   
       \KwIn{ Good $j\in [m]$.} 
       \For{$i\in [n]$}{
            $p^{(k)}_j \leftarrow (1+\epsilon)p_j$\label{SRpriceIncreaseNext}\;
            \lIf{$(1+\epsilon)p_j > 1$}{$c^{(i)}\leftarrow c^{(i)}/(1+\epsilon)$}
            }
        $p_j \leftarrow (1+\epsilon)p_j$; recalculate $s^r_i$ for all agents; $a_j = e_j \cdot \min\{1, 1/p_j\}$; \,$l_j\gets a_j$\,  \label{SRpriceIncrease}      
\caption{RaisePrice($j$) \label{SRprocRaiseprice}}
\end{procedure}

\subsection{Description of the algorithm}
\paragraph{Initialization}
In the case of exchange markets, we used Assumption~\ref{assumption1} to state that approximate equilibrium prices $\ge \mathds{1}$ exist, and then we were able to initialize the algorithm by setting all prices to $1$.
This is not viable for Fisher markets, where even the total budget might be smaller than $m$. 
Instead, we assume that we are given some initial, small enough prices
$\bar p< \1$  and optimal bundles $x^{(i)}\in D_i(\bar p, b_i)$ such that
$\sum_{i=1}^n x^{(i)} \ge e$ and initialize 
$p^{(i)} = \bar p$ for all $i$, and set all $c^{(i)}$'s so that  $c^{(i)} \le x^{(i)}$ and $\sum_i c^{(i)} = e$.
For example, this is achieved whenever there is a single agent that overdemands all the goods under arbitrary low prices, i.e.,
there is $i \in [n]$ and $\bar p<\1$ such that $x^{(i)} \ge e$ for $x^{(i)} \in D_i(\bar p, b_i)$. 

\begin{remark}\label{remark:emptyInitialization}
A simple alternative initialization is to set the price of good $j$ as $p_j=\frac{\epsilon}{e_j\cdot m}\sum_{i} b_i$, and start with allocations $c^{(i)}=0$. The drawback is that we would obtain a slightly weaker equilibrium at termination. Part~\ref{def:SRno-over} of Definition~\ref{def:approx-SReq} requires that the available amount of each good is fully sold; we would need to weaken this property to say that the total price of all unsold goods would be $\le\epsilon \sum_{i} b_i$. Below, we describe the analysis for the case where initially all goods are fully sold, but it can be easily adapted to this version.
\end{remark}

For each good $j=1,2,\ldots,m$, we partition the available amount as $a_j=l_j+h_j$ according to the price it is sold at:
amount $l_j$ is sold at the lower price $p_j$; and
amount $h_j$ is sold at the higher price $(1+\epsilon)p_j$.
We only explicitly maintain $l_j$ in the algorithm.
We further maintain $l_j > 0$ at the beginning of every step, i.e., there is always a part of the good
that is unsold or owned by an agent at the lower price. 

\paragraph{The Outbid subroutine}
As before ~\pOutbid$(i,j,\gamma)$,
controls how the ownership of goods may change. 
When agent $i$ would like to purchase an additional amount $\gamma$ of good $j$,
we consider agents $k$ one-by-one who are paying the lower price $p_j$ for good $j$, i.e., $j\in L_k$.
Agent $i$ may take over some of this amount by offering a higher price  $(1+\epsilon)p_j$.

\paragraph{The RaisePrice subroutine}
Procedure~\pRaisePrice$(j)$ is called when $l_j=0$ for a good $j$, i.e., it is only sold at the higher price $(1+\epsilon)p_j$. 
In this case, we increase the market price to $(1+\epsilon)p_j$, set all individual prices $p_j^{(k)}$ to this value, and 
if $(1+\epsilon)p_j> 1$ we decrease the value $c^{(i)}_j$ of each agent $i$ by factor $(1+\epsilon)$.
We also recalculate the relative surpluses of the agents using the definition and set $l_j=a_j$.

\paragraph{Steps}
The algorithm terminates once as soon as the
total relative surplus drops below $3\epsilon \pr{p}{e}$.
We consider agents $i$ with $s^r_i>0$ one-by-one and offers them a chance to spend more money on the goods by outbidding.
The outbidding proceeds same as in the case of exchange markets by using the \pFindNewPrices{} subroutine.  
Similarly as before, this process results in $l_j = 0$ for good $j$, then we call \pRaisePrice$(j)$ to increase the market price.

\subsection{Analysis}
As previously mentioned, an ($\epsilon$-)SR equilibrium may not exist
at all. In such cases, our algorithm will never terminate, increasing
the prices unlimitedly. 
We give the running time in terms of the
ratio $p_{SR\max}/p_{\min}$. 
Here, $p_{\min}=\min_j \bar p_j$, is the
smallest one among the initial prices, and $p_{SR\max}$ is an upper
bound on the prices in the algorithm; note that we may have $p_{SR\max}=\infty$.
In Section~\ref{section:price-bound}, we give a bound in terms of the
maximum and minimum values of the partial derivatives of the utility function.

\begin{theorem}\label{thm:SRrunning-oracle}
Let 
$T_F$ be an upper bound on the running time of the
subroutine \pFindNewPrices. Then there exists an auction 
algorithm that finds a {$4\epsilon$-approximate SR}
equilibrium in time 
\[
\displaystyle O\left(\frac{nmT_F}{\epsilon^2}\log \left(\frac{p_{SR\max}}{p_{\min}}\right)\right).
\]
\end{theorem}

\begin{lemma}\label{lem:SRinvariants}
If all agents have WGS demand systems, then the invariants
\ref{SRprop:price-inc}-\ref{SRprop:pay}  hold after every step.
\end{lemma}
\begin{proof}
\begin{description}
\item[~\ref{SRprop:price-inc}] This is immediate.
\item[~\ref{SRprop:fully-sold}] The algorithm maintains  $a_j=l_j+h_j= e_j \cdot \min\{1, 1/p_j\}$ for all goods.
This is guaranteed by property~\ref{cond:still-buy} of the procedure
\pFindNewPrices{}, the fact that $c^{(i)}_j$ may only decrease if
another $c^{(k)}_j$ increases by the same amount, and by the procedure \pRaisePrice{} if $p_j$ becomes more than $1$.
\item[~\ref{SRprop:individual}] 
The bounds $p\le p^{(i)}\le (1+\epsilon)p$ are immediate from condition~\ref{cond:higher-price} 
in \pFindNewPrices{} and from the procedure \pRaisePrice{}.
\item[~\ref{SRprop:dominate}]
Suppose these properties hold for every agent before a step of agent $i$. 
The requirements~\ref{cond:still-buy} and \ref{cond:higher-price} guarantee that $c^{(i)}$ is dominated by a bundle $x^{(i)}\in D_i(p^{(i)}, b_i)$ and prices satisfy $p\le p^{(i)}\le (1+\epsilon)p$, for each agent $i$. 

Now, consider an agent $k$ different from $i$. 
In the step, $k$ could lose a part of a good only through the
outbid and hence $c^{(k)}$ does not increase. As long as the prices $p^{(k)}$
do not change, \ref{SRprop:dominate} holds trivially.
The only time $p^{(k)}$ can change is the price increase step in line~\ref{SRpriceIncreaseNext}, namely, if
$p_j$ increases to $(1+\epsilon)p_j$, it forces
$p^{(k)}_j=(1+\epsilon)p_j$. Note that the price increase only happens once $l_j=0$.
Assume we had $p^{(k)}_j<(1+\epsilon)p_j$ before the price increase,
that is, agent $k$ was buying good $j$ at the lower price $p_j$. By
$l_j=0$ and invariant \ref{SRprop:pay}, it follows that $c^{(k)}_j=0$ at this point. 
The WGS property implies that after increasing $p^{(k)}_j$, the bundle $c^{(k)}$ will
be still dominated by an optimal bundle.
\item[~\ref{SRprop:pay}] It is straightforward to check that the form of the relative surplus is maintained. \qedhere
\end{description}
\end{proof}

The bound on the number of rounds two consecutive calls to \textnormal{\pRaisePrice{}} is the same as the one for Algorithm~\ref{exchangeMarkets}. They differ slightly due to using Assumption~\ref{assumption:budgetIncreases} instead of Assumption~\ref{non-satiation}. 

\begin{lemma}\label{lemma:SRroundsInIteration}
The number of rounds between two consecutive calls to \textnormal{\pRaisePrice{}} is at most $ 2/\epsilon$.
\end{lemma}

\begin{proof}
Let $p$ be the market prices after a call to \pRaisePrice{}, and consider a sequence of steps with the same market prices.
Consider a step of an agent $i$ during this sequence.
If $i$ buys $l_j$ of a good $j$, then \pRaisePrice{} is called and the sequence finishes. 
Thus, we can assume that during this sequence, every agent $i$ gets the amount of each good they desire.

Let $\varphi$ denote the total amount of money spent at a certain point
of this sequence of steps that is spent by the agents on higher price
goods. That is, 
\[
\varphi=(1+\epsilon)\sum_{i=1}^n \sum_{j\in H_i} c^{(i)}_j p_j.
\]

\begin{claim}
Let $s^r_i$ denote the relative surplus of agent $i$ at the beginning of their step.
Then  the value of $\varphi$ increases by at least $s^r_i-2.25\epsilon b_i$ during agent $i$’s step.
\end{claim} 

\begin{claimproof}
Recall Cases 1-3 in the description of the step. Let $T_k$ be the set of
goods that fall into case $k$, that is, $T_1\cup T_2\cup T_3=[m]$.
\begin{itemize}
\item  
If $j\in T_1$, then $(1+\epsilon)p_jy_j$ will be added to
$\varphi$ in the \pOutbid{} subroutine: In this case, the agent also
outbids itself, moving  good $j$ from $L_i$ to $H_i$.
\item If $j\in T_2$, then $(1+\epsilon)p_j( y_j-c^{(i)}_j)$
  will be added to $\varphi$ in the \pOutbid{} subroutine.
\item
If $j\in T_3$, then we do not increase
$\varphi$. Nevertheless, \ref{cond:higher-price} guarantees
that 
$\tilde  p_j (y_j- c^{(i)}_j)\le \epsilon \tilde p_j c^{(i)}_j.$
Consequently,
\begin{equation}\label{eq:SRT3}
\sum_{j\in T_3}\tilde  p_j (y_j- c^{(i)}_j)\le \epsilon \pr{\tilde p}{c^{(i)}}\, .
\end{equation}
\end{itemize}
Also note that $\tilde p_j=(1+\epsilon)p_j$ if $j\in T_1\cup T_2$.
Let $\Delta\varphi$ denote the increment in
$\varphi$; this can be lower bounded as
\[
\begin{aligned}
\Delta\varphi&=\sum_{j\in T_1} \tilde p_j y_j+\sum_{j\in T_2} \tilde
p_j(y_j-c^{(i)}_j)
= \pr{\tilde p}{y} -  \sum_{j\in T_3} \tilde p_j y_j - \sum_{j\in T_2} \tilde
p_jc^{(i)}_j \\
&\ge  \pr{\tilde p}{y} - \sum_{j\in T_3} \tilde p_j (y_j - c^{(i)}_j) - \pr{\tilde p}{c^{(i)}}
\ge \pr{\tilde p}{y} -(1+\epsilon) \pr{\tilde p}{ c^{(i)} }\enspace ,
\end{aligned}
\] 
using \eqref{eq:SRT3}.
The money spent by the agent at the beginning of the step is
$\pr{p^{(i)}}{x^{(i)}} - s^r_i$. Good  $j$ is purchased at price
at least $p_j$ according to \ref{prop:pay}, and $\tilde p_j\le
(1+\epsilon) p_j$. 
Consequently, $\pr{\tilde p}{c^{(i)}} \le (1+\epsilon) (\pr{p^{(i)}}{x^{(i)}}  - s^r_i)$.
Assumption~\ref{assumption:budgetIncreases} yields $\pr{\tilde
p}{y}\ge \pr{p^{(i)}}{x^{(i)}} $. Therefore, using 
$\epsilon<0.25$, we obtain
\[
\Delta\varphi\ge \pr{\tilde p}{y} - (1+\epsilon)^2 \left(\pr{p^{(i)}}{x^{(i)}}  - s^r_i\right) 
\ge s^r_i + \pr{\tilde p}{y} - (1+\epsilon)^2\pr{p^{(i)}}{x^{(i)}} 
\ge s^r_i - 2.25 \epsilon  \pr{\tilde p}{y} 
\ge  s^r_i - 2.25 \epsilon b_i\,.
\]
In the last inequality we used $\pr{p}{y}\le b_i$. This completes the proof.
\end{claimproof}

As long as $\sum_{i=1}^n s_i^r>3\epsilon \sum_{i=1}^n b_i$, the claim guarantees that $\varphi$ increases in every round by at least 
\[3\sum_{i=1}^n b_i - 2.25\epsilon \sum_{i=1}^n b_i = 0.75 \epsilon \sum_{i=1}^n b_i\,.\] 
Since
$\varphi\le \sum_{i=1}^n b_i$,
the number of rounds until the next call the \pRaisePrice{} is bounded
by $2/\epsilon$.
\end{proof}

We need one more lemma to bound the total running time of \pOutbid{}.
\begin{lemma}\label{lem:SRwhile-bound}
Between two calls to \textnormal{\pRaisePrice{}}, for every agent $k\in[n]$ and good $j\in [m]$, there can be only one occasion that during a call to \textnormal{\pOutbid{($i,j,\gamma$)}} for some $i\in[n]$ and $\gamma>0$, $c_j^{(k)}$ is set to 0 in line~\ref{SRline:outbid-zero}.
\end{lemma}
\begin{proof}
Assume  $c_j^{(k)}$ was set to 0 in a certain call to \pOutbid.
Before it can be set to 0 again,  agent $k$ must have obtained a positive amount of good $j$. This could only happen in a call to \pOutbid{($k,j,\gamma$)} for some $\gamma>0$. However, until the next call to \pRaisePrice{}, agent $k$ may only buy good $j$ at the higher price $(1+\epsilon)p_j$, and henceforth $j\in H_k$. Thus, agent $k$ cannot be selected again in line~\ref{SRline:select-agent} at a call to \pOutbid{($i,j,\gamma$)}.
\end{proof}

\begin{proof}[Proof of Theorem~\ref{thm:SRrunning-oracle}]
Lemma~\ref{lem:SRterm} shows that at termination, the algorithm returns a $4\epsilon$-market equilibrium.

Let us start by bounding the total running time between two calls to \pRaisePrice{}. By Lemma~\ref{lemma:SRroundsInIteration}, there are at most $2/\epsilon$ rounds. Every round comprises $n$ steps, and every steps calls the procedure \pFindNewPrices{} exactly once.
Therefore, the time taken by \pFindNewPrices{} during this sequence of steps is $O(nT_F/\epsilon)$.

The total number of calls to \pOutbid{} is 
$m$ in each step, totaling to $O(nm/\epsilon)$.
Let us bound the number of repeats in the `while' loop (lines~\ref{SRl:while-first}--\ref{SRl:while-last}) in all calls to \pOutbid{} between two calls to \pRaisePrice{}
In a call to  \pOutbid{($i,j,\gamma$)}, in all but the final call to the `while' loop, we set $c^{(k)}_j=0$ for some agent $k$. 
By Lemma~\ref{lem:SRwhile-bound}, the total number of these events is at most $O(nm)$.
Hence, the number of repeats in the `while' loop between two calls to \pRaisePrice{} is  $O(nm/\varepsilon+nm)=O(nm/\varepsilon)$. Each repeat takes $O(1)$ time.

From the above, the total time of the \pOutbid{} calls is $O(nm/\epsilon)$ between two calls to \pRaisePrice{}. A call to  \pRaisePrice{} takes $O(nm)$ time.
Consequently, the
total time of such a sequence of steps is $O(nT_F /\epsilon+nm/\epsilon) = O(nT_F /\epsilon)$, using the assumption that $T_F = \Omega(m)$.

By definition $p_{\min}$ is the minimum price at the initialization and $p_{SR\max}$ is the maximum price of any good reached in the algorithm. Consequently,  \pRaisePrice{} can be called at most $O(m \log_{1+\epsilon}(p_{SR\max}/p_{\min}))=O(\frac m\epsilon \log (p_{SR\max}/p_{\min}))$ times. The claimed bound follows.
\end{proof}

\subsection{Conditions on the existence of SR-equilibria}\label{section:price-bound}
We now present a general bound on the value of $p_{SR\max}$.
Suppose that the demand system of each agent $i$ is provided in terms
of a monotone concave and differentiable utility function $u_i$ in the form
\eqref{prog:optimalBundle}. We now assume that each $u_i$ is
differentiable. The arguments here can be easily adopted for the
non-differentiable case by using subgradients.
We let 
\begin{equation}\label{eq:v-max}
\begin{aligned}
D:=\frac{\max_i b_i}{p_{\min}},\quad v_{i \max} := \max_{j} \partial_j u_i(0), & \quad v_{i \min} := \min_{j} \{ \partial_j u_i(D\cdot \1): \partial_j u_i(0)>0\},\\
 & \quad V_{\max} := \max_i \frac{v_{i\max}}{v_{i\min}},\\
 & \quad e_{\max} := \max_j {e_j}.
\end{aligned}
\end{equation}
Note that if $\partial_j u_i(0)=0$, then agent $i$ is not interested
in good $j$ at all.  In case $\partial_j u_i(0)>0$ we say that agent
$i$ {\em is interested} in good $j$. 
 Note that $D$ is an upper bound on the
amount of any single good that any agent could buy.

%We note that $t_{\max}=\infty$
%could be possible. However, we can truncate the value of every $1$
%to $\min\{1,\sum_i b_i/e_j\}$ without changing the problem, since the total spending is at most the total budget; the 
%price of a good can never rise above this value in the algorithm or in an
%SR-equilibrium. 
%Thus, we may assume $t_{\max}\le \sum_i b_i/e_j$ in the
%bounds below.

\paragraph{A necessary condition for the existence of SR-equilibria} 
The condition $\sum_i b_i \le
\sum_j e_j $ is necessary for the existence of an
SR-equilibrium, since $\sum_j e_j$ is the total amount of money that
can be spent on the goods. One can formulate an extension of this,
that amounts to Hall's condition in a certain graph. Let $(A\cup G,E)$
denote the bipartite graph where the two classes $A$ and $G$ represent
the agents and goods, respectively, we add an edge $(i,j)\in E$ if $\partial_j
u_i(0)>0$, that is, if agent $i$ is interested in good $j$. For a
subset $S\subseteq A$, we let $\Gamma(S)\subseteq G$ denote the set of
neighbors in this graph. Then, Hall's condition, that is, 
\begin{equation}\label{eq:Hall}
\sum_{i\in S} b_i \le \sum_{j\in \Gamma(S)} e_j, \quad \forall
S\subseteq A
\end{equation}
 is a necessary
condition on the existence of an SR-equilibrium. Note that
this condition is not sufficient: it holds for the example of
Cobb-Douglas utilities, 
where no SR-equilibrium exists, as explained before Definition~\ref{def:approx-SReq}.

\paragraph{Upper bounds on the prices} 
We now give a bound on $p_{SR\max}$ in terms of $V_{\max}$
and $e_{\max}$. We first consider the case when every agent is
interested in all goods. 
In this case, \eqref{eq:Hall} reduces to the
case when $\Gamma(S)$ contains all goods.
Note that the bounds are finite only if $v_{i\min}>0$, and $v_{i\max}$
is finite. For the Cobb-Douglas utilities, $v_{i\max} = \infty$.

\begin{restatable}{lemma}{pSRmaxBound}\label{lemma:pSRmaxBound}
Assume the demand systems of the agents are given in form
\eqref{prog:optimalBundle} for  monotone concave and differentiable utility functions $u_i$.
\begin{enumerate}[label=(\roman*)]
\item  Suppose that every agent is interested in every good, that is,
  $\partial_j u_i(0)>0$ for every agent $i$ and every good $j$. Assume that 
$\sum_i b_i \le \sum_j e_j $. Then, the prices throughout
the auction algorithm remain bounded by $(1+\epsilon)^2 e_{\max} V_{\max}$.
\item Assume condition \eqref{eq:Hall} holds with strict inequality
  for all $S\subseteq A$. Then, the prices throughout
the auction algorithm remain bounded by $(1+\epsilon)^{n} e_{\max} V_{\max}^{n-1}$.
\end{enumerate}
The same bounds are valid for any $\epsilon$-SR equilibrium.
\end{restatable}
\begin{proof}
Let us first consider {\em (i)}.
Let $p$ denote the market prices at a certain point of the algorithm,
or at an $\epsilon$-SR equilibrium, 
and  let $p_{SR\min}$ be the minimal price among those. 
Observe that this might be different from $p_{\min}$, since $p_{\min}$ is the minimal price at initialization. Let
$\ell$ be a good with $p_\ell=p_{SR\min}$. 

We use the KKT conditions of convex
program~\eqref{prog:optimalBundle}. We let $\beta^{(i)}$ denote the
Lagrange multiplier of the budget constraint for agent $i$.
Then, $\partial_j u_i(x^{(i)}) \le \beta^{(i)} p^{(i)}_j$ for all goods $j$; and equality holds whenever $x^{(i)}_j > 0$.
Recall that each good $j$ is owned by some agent
during the algorithm as well as in an $\epsilon$-SR-equilibrium. 

Consider a good $j$, and let $i$ be an agent buying $j$, i.e.,
$c^{(i)}_j > 0$ and therefore $x^{(i)}_j > 0$.
By the above, 
${p^{(i)}_j}/{p^{(i)}_\ell} \le {\partial_j
  u_i(x^{(i)})}/{\partial_\ell u_i(x^{(i)})}$. 
The assumption that every
agent is interested in every good means that  $v_{i \min} =
\min_{j} \partial_j u_i(D\cdot \1)$.
Since $x^{(\ell)}\le D\cdot \1$, concavity implies $\partial_\ell
u_i(x^{(i)})\ge v_{i\min}$. 
We also get $\partial_j u_i(x^{(i)})\le v_{i \max}$. 
Consequently, ${p^{(i)}_j}/{p^{(i)}_l} \le {v_{i\max}}/{v_{i\min}} \le V_{\max}$.
Finally, since $p \le p^{(i)} \le (1+\epsilon)p$, we have ${p_j} \le (1+\epsilon) p_{SR\min} V_{\max}$ for any good $j$. 

The proof is completed by showing that $p_{SR\min}\le
(1+\epsilon)e_{\max}$. 
To prove this, we first show that once $p\ge e$, the algorithm
terminates. Indeed, if $p\ge e$, then the agents spend $\sum_j e_j$ in
total, since the amount $a_j = e_j \cdot \min\{1, 1/p_j\}$ is always fully sold. 
The condition $\sum_i b_i \le \sum_j e_j$ shows that agents
cannot have any surplus at this point. 
Thus, once the lowest price rises above
$e_{\max}$, the algorithm terminates. Since the prices increase in steps of
$(1+\epsilon)$, we get that $p_{SR\min} \le (1+\epsilon)e_{\max}$. 

\medskip

Let us now consider part {\em (ii)}. We take the bipartite graph
$(A \cup G,E)$, and on the same set of nodes we define a directed graph
as follows. We orient all edges in $E$ from $A$ to $G$, and 
also add the arc $(j,i)$ whenever $x_j^{(i)}>0$. Fix any good $j$, and
let  $S$ be the set of agents in
$A$ reachable from $j$ in this directed graph. 
Note that the set of goods reachable from $j$ will be precisely
$\Gamma(S)$. Let $\ell\in \Gamma(S)$ be the good with the
lowest price $p_\ell$. As above, we can show that $p_\ell\le
(1+\epsilon)e_{\max}$, since $p\ge e$ is not possible. Indeed, once
$p\ge e$, then all the available amounts of goods in $\Gamma(S)$ are
fully sold, and their total value is $\sum_{j\in \Gamma(S)} e_j >
\sum_{i\in S} b_i$ by the assumption. By the definition of $S$, no
agent outside $S$ pays for goods in $\Gamma(S)$, leading to a contradiction.

The directed graph contains a path of length $\le 2(n-1)$ from $p_j$
to $p_\ell$. As in the proof of part {\em (i)}, one can argue that for
any two consecutive  goods $j'$ and $j''$ %\jnote{$j$ and $j'$?} 
on this path, $p_{j'}/p_{j''}\le (1+\epsilon) V_{\max}$. This implies the bound.
\end{proof}

\paragraph{Bounding the prices for Gale demand systems}
Consider now the demand system $G^{u_i}(p,b_i)$ defined from a monotone
concave utility function by \eqref{galeObjective}.
The concavity implies that they will never spend more than $b_i$ in the optimal bundle, 
and moreover agent $i$ might spend strictly less than $b_i$.
Thus, even if $\sum_i b_i\le \sum_j e_j$ does not hold an SR-equilibrium could exist.

Still, we can obtain the same bounds as in
Lemma~\ref{lemma:pSRmaxBound} on the prices. The proof is identical,
noting that the KKT conditions for \eqref{galeObjective} also imply ${p^{(k)}_j}/{p^{(k)}_\ell} \le {\partial_j
  u_k(x^{(k)})}/{\partial_\ell u_k(x^{(k)})}$ if $x^{(k)}_j>0$, and
the fact that agent $i$ spends at most $b_i$ in their optimal bundle.

\begin{restatable}{lemma}{pSRmaxBound-Gale}\label{lemma:pSRmaxBound-Gale}
Assume {every agent has a  Gale demand system
\eqref{galeObjective}} for  monotone concave and differentiable utility functions $u_i$.
\begin{enumerate}[label=(\roman*)]
\item  Suppose that every agent is interested in every good, that is,
  $\partial_j u_i(0)>0$ for every agent $i$ and every good $j$. Assume that 
$\sum_i b_i \le \sum_j 1 $. Then, the prices throughout
the auction algorithm remain bounded by $(1+\epsilon)^2 e_{\max}
V_{\max}$.
\item Assume condition \eqref{eq:Hall} holds with strict inequality
  for all $S\subseteq A$. Then, the prices throughout
the auction algorithm remain bounded by $(1+\epsilon)^{n} e_{\max} V_{\max}^{n-1}$.
\end{enumerate}
The same bounds are valid for any $\epsilon$-SR equilibrium.
\end{restatable}

%%% Local Variables:
%%% mode: latex
%%% TeX-master: "main"
%%% End:

%% file: basplc.tex
%!TEX root = mainPlain.tex
\section{Approximating Nash social welfare}\label{section:BASPLC}
As an application of the spending restricted auction algorithm in Section~\ref{section:spendingRestricted}, we give a polynomial-time $(2{e^{1/2e}}+\epsilon) \approx 2.404$-approximation algorithm 
for the NSW problem under capped separable piecewise linear concave (SPLC) utilities---the common generalization of the models in
\cite{anari2018nash} and \cite{garg2018approximating}.
We consider an instance of the NSW problem with $n$ agents and $m$ goods, in which we have $D_{j}$ units (copies) of good $j$.
Each agent $i$ has a capped SPLC utility function defined as follows (see Figure~\ref{figure}). 
For every good $j$, agent $i$ has $k_{ij}$
{\em segments} with \emph{strictly} decreasing utility rates
$u_{ij1}>u_{ij2}>\ldots>u_{ijk_{ij}}\ge 0$. 
Segment $t\in [k_{ij}]$ has length $d_{ijt}$ and agent $i$ values at $u_{ijt}$ each of the units in the segment.
Each variable $d_{ijt}$ takes positive integral values.
We assume that $\sum_{t\in[k_{ij}]} d_{ijt} = D_j$.
Furthermore,  agent $i$'s utility is capped at $U_i$, i.e., their utility is the minimum of $U_i$ and the sum of the
utilities accumulated from the goods.

\begin{figure}
  \centering
  \includegraphics[width=0.5\textwidth]{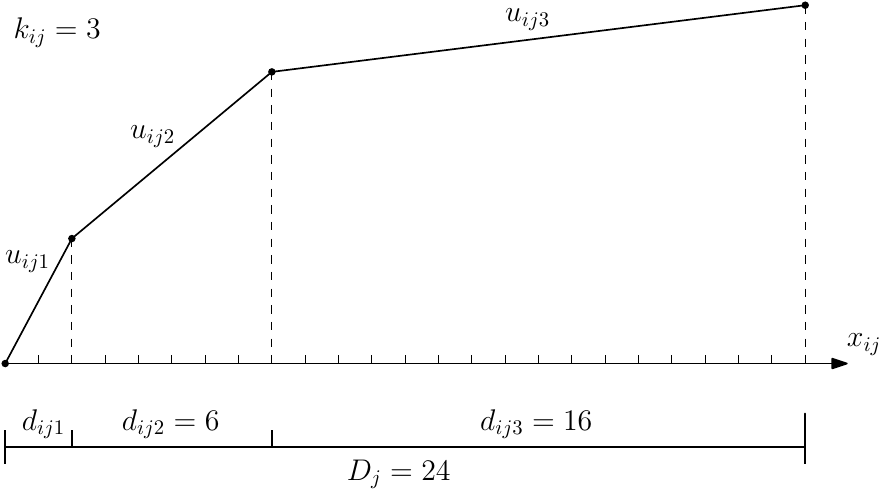}
  \caption{Agent $i$'s utility for good $j$.}
  \label{figure}
\end{figure}

Chaudhury et al.~\cite{ChaudhuryCGGHM18} gave a $e^{1/(1+\epsilon)e} \approx 1.44$-approximation algorithm for the problem,
while Anari et al.~\cite{anari2018nash} studied the problem with SPLC utilities ($U_i = \infty$) and gave a $2$-approximation algorithm. 
The running times of these algorithms depend linearly on $M$, where $M=\sum_{j\in[m]} D_j$. 
In other words, \cite{anari2018nash}  and~\cite{ChaudhuryCGGHM18} use segments of length 1.
Therefore, when multiple copies of a good have the same utility rate, their algorithms run in pseudopolynomial time. 
Using the auction algorithm, we give an approximation algorithm running in polynomial time:  the utility function is specified by the utility rate and the length of a segment rather than $d_{ijt}$ segments of length one with the same utility.
The approach consists of three parts:
\begin{itemize}
   \item Finding an SR-equilibrium for the instance of Fisher market arising as a relaxation of the NSW problem.
    As already mentioned in the introduction, the natural relaxation of the NSW problem uses the SR-equilibrium with respect to the Gale demand system, where each agent has budget $1$. We use the auction algorithm to find such an approximate SR-equilibrium $(x, p)$. 
    It is worth pointing out that this is the main reason why we obtain a better running time guarantee than the existing approaches.
    \item Upper bounding the optimal value of the NSW in terms of prices $p$.
    \item Rounding the allocation $x$.
 \end{itemize} 
The last two rely on the ideas originally given by Cole and Gkatzelis~\cite{cole2015approximating} and extended in~\cite{anari2018nash,garg2018approximating}.
More precisely, for the upper bound we follow~\cite{anari2018nash} and we explain how the rounding reduces to the case of capped linear utilities~\cite{garg2018approximating}.
For the sake of simplicity, we present an upper bound and the rounding for an exact SR-equilibrium similarly to the one in~\cite{garg2018approximating}.
The modification to an approximate SR-equilibrium is straightforward.
For the upper bound and rounding we make the assumption that $u_{ijt} \le U_i$, 
as we could redefine the utilities to $u_{ijt} \leftarrow \min\{u_{ijt}, U_i\}$ without changing 
the objective value of the feasible allocations for the NSW instance.

\subsection{SR equilibrium under Gale demand systems of a  capped SPLC}
\label{section:basplcEquilibrium}
We now consider the Gale demand system for {\em  capped SPLC}.
We first show that the corresponding demand system is WGS---thus we can use the auction algorithm; 
and then we give an implementation of the \pFindNewPrices{} subroutine for this demand
system. Note that the convex programming approach does not immediately
apply, since the utility function is not differentiable, and the
optimal bundle is not unique. Instead, we give a simple price
increment procedure, an extension of that in
Lemma~\ref{lemma:linearPrices} for linear utilities.
As both the WGS property and \pFindNewPrices{} refer to a fixed agent, 
we drop the term $i$ denoting the agent in the subscripts.

The Gale demand system $G^u(p,b)$ is defined as the set of optimal
solutions to the following formulation.\footnote{It can be easily verified, 
using the KKT conditions given below, that {\em admissible spendings} in 
\cite{anari2018nash} correspond to the case when $U=\infty$, and {\em
  modest and thrifty} demand bundles in \cite{garg2018approximating}
to the case when $k_j=1$ for all $j$ with $d_{j1}=\infty$.}
\begin{equation}\label{prog:SPLCcaps}
\begin{aligned}
\max~~ b\log \left(\sum_j \sum_t x_{jt} u_{jt} \right) &- \sum_j p_j \sum_{t = 1}^{k_j} x_{jt}\\
\mbox{s.t.} \qquad x_{jt} &\le d_{jt} \qquad \forall j \in [m],t \in [k_j]\\
\sum_{j=1}^{m} \sum_{t=1}^{k_j} x_{jt} u_{jt} &\le U \\
x &\ge 0 \,. 
\end{aligned}
\end{equation}

Let us now present the KKT conditions characterizing the optimal
solution $x^*$.
Let $r_{jt}$ be the Lagrange multipliers of the constraint $x_{jt} \le
d_{jt}$ and $\gamma$ the Lagrange multiplier of the utility constraint.
Recall that $u(x^*) = \sum_j \sum_t u_{jt} x^*_{jt} $. 
We have the following:
\begin{enumerate}[label=(\roman*)]
\item\label{splcKKTi} $\frac{b u_{jt}}{u(x^*)} \le r_{jt} +  p_j  + u_{jt}\gamma$,
\item\label{splcKKTii} $\frac{b u_{jt}}{u(x^*)} = r_{jt} + p_j +
  u_{jt}\gamma$ whenever $x^*_{jt} > 0$,
\item\label{splcKKTiii} $x_{jt}^* =d_{jt}$ whenever $r_{jt} > 0$, and
\item\label{splcKKTiv} $\sum_j \sum_t x^*_{jt} u_{jt} = U$ whenever $\gamma>0$.
\end{enumerate}

\begin{lemma}[WGS property]
The Gale demand system for  capped SPLC utilities satisfies the WGS property.
\end{lemma}
\begin{proof}
Let us consider prices $p'$ defined by $p'_j = p_j$ for $j \in
[m]\setminus \{\ell\}$ and $p'_{\ell} > p_{\ell}$. We show
that there is an optimal bundle $x'$ at prices $p'$ such that \b{
$x'_{jt}\ge x_{jt}^*$} for all $j\neq \ell$ and all $t\in [k_j]$.
For prices $p'$, let $\overline u$ be the optimal utility in~\eqref{prog:SPLCcaps} and 
let $\gamma'$ be the Lagrange multiplier for the constraint on the maximum utility achieved.
We consider two cases.

\textit{Case 1: $\bar u < u(x^*_i)$.}
By~\ref{splcKKTii}, $x^*_{jt} > 0$ implies $\frac{u_{jt}}{p_j} \ge \frac{u(x^*)}{b}$.
Thus, we have $\frac{u_{jt}}{p'_j} = \frac{u_{jt}}{p_j} \ge \frac{u(x^*_i)}{b} >
\frac{\overline u}{b}$ for all $j, t$ with $x^*_{jt} > 0$ and $j\neq \ell$.

Moreover, by~\ref{splcKKTii} and~\ref{splcKKTiii}, 
if $\frac{u_{jt}}{p'_j}> \frac{\overline u}{b}\cdot
\left(1+\gamma'\cdot \frac{u_{jt}}{p'_j}\right)$ then $x'_{jt} =d_{jt}$.
By~\ref{splcKKTiv}, $\bar u < u(x^*_i) \le U$ implies that $\gamma' = 0$, 
and hence $x'_{jt}=d_{jt}$ for all $j, t$ with $x^*_{jt} > 0$ and $j\neq \ell$.
In other words, for every good $j$, $j\neq \ell$, every segment of the good that the agent was buying at prices $p$ 
is fully bought at prices $p'$. 
The lemma follows.

\textit{Case 2: $\bar u = u(x^*_i)$.}
It suffices to prove that the optimal solutions of the following knapsack linear program satisfy the WGS property.
\begin{equation}
\begin{aligned}
\min~~ \sum_j p_j \sum_{t = 1}^{k_j} & x_{jt}\\
\mbox{s.t.} \qquad x_{jt} &\le d_{jt} \qquad  \forall j \in [m],t \in [k_j]\\
\sum_{j=1}^m \sum_{t=1}^{k_j} x_{jt} u_{jt} &= \bar u\\
x &\ge 0 \,. 
\end{aligned}
\end{equation}
Suppose that the optimal solution $x$ is unique, then it can be build in a greedy fashion. 
Order the segments of all goods in a decreasing order of the fractions
$\frac{u_{jt}}{p_j}$. 
Then $x$ is obtained by purchasing the segments 
(i.e. allocating $x_{jt}=d_{jt}$) in the above order until the utility becomes $\bar u$; 
having in mind that the last purchased segment might be purchased only partially.

To prove the WGS property we consider increasing price $p_\ell$ of a good $\ell$.
The price increase will cause the segments corresponding to good $\ell$
 to move further back in the ordering while the relative order of all
rest of the segments remains unchanged. Hence,
by the greedy argument above, one can find
an optimal solution $x'$ with $x'_{jt}\ge x_{jt}$ for all $j\neq \ell$
and $t\in [k_j]$.

In the case there are multiple optimal solutions, a similar argument holds 
since two optimal solution 
differ only on a set of goods with the same ratio $\frac{u_{jt}}{p_j}$.
\end{proof}

As in Section~\ref{section:gale}, we show that the following slightly
more general version of  \pFindNewPrices{} can be implemented.
Let  $p,q,c \in \R^m_+$ and $x \in G^u(p, b)$ such that $p \le q$ and $c
\le x$. Find $\tilde p$ and $y$ such that
\begin{enumerate}[label=(\Alph*')]
\item\label{firsts} $y \ge c$ where $y \in G^u(\tilde p, b)$, and
\item\label{seconds} $p\le \tilde p \le q$ and $\tilde p _j = q_j$ whenever $y_j > c_j$.
\end{enumerate} 

\begin{lemma}[\pFindNewPrices{}]
\label{lemma:findNewPricesBASPLC}
The procedure \pFindNewPrices{} can be implemented in time $O(K)$
for Gale demand systems with capped SPLC utilities, where $K = \sum_{j 
\in [m]} k_j$ is the number of segments with \emph{different} marginal utility. 
%\enote{I added $O(K)$ instead of "polynomial as was before".} \jonte{Yes, this is fine.}
\end{lemma}
The proof is via an algorithm that is an extension of the one in the
proof of Lemma~\ref{lemma:linearPrices} for linear utilities.
\begin{proof}
We present an algorithm for finding such prices $\tilde p$ and bundle $y$. 
The algorithm initializes $\tilde p = p$ and $y = c$. The prices as
well as the allocations are non-decreasing throughout the
algorithm. Note that $u(y)<U$ at the initialization; otherwise, $c=x$ would
follow and we can simply output $y=x$ and $\tilde p=p$.
 We maintain $p\le \tilde p \le q$ throughout.
%, and let $Q\subseteq [m]$ denote the set of goods $j$ with $\tilde
%p_j=q_j$. 
For each $j\in [m]$, let $t_j\in[k_j]$  denote
the first segment of a good $j$ that is not completely sold in $y$,
i.e., the minimal $t_j$ such that $y_{jt_j} < d_{jt_j}$. We call this
the {\em active segment} for $j$.

Consider the optimal bundle $x$ such that
$c\le x$, and let $\gamma$ be the Lagrange multiplier for the utility cap
constraint for $x$. We initialize $\beta= \left(b/u(x)
  -\gamma\right)^{-1}$. Then, from \ref{splcKKTi}-\ref{splcKKTiii}
we see that if $x_{jt}=0$ then $u_{jt}/p_j\le \beta$, if $0<x_{jt}<d_{jt}$ then  
 $u_{jt}/p_j=\beta$, and if 
$x_{jt}= d_{jt}$ then  $u_{jt}/p_j\ge \beta$. 

\paragraph{Stage I: enforcing the complementary conditions}
The algorithm proceeds in two stages. In the first stage, 
we consider
the goods for which $u_{jt_j}/\tilde p_j>\beta$ yet
$y_{jt_j}<d_{jt}$. (Recall that we initialized $y=c$ and $\tilde p=p$.)
For each such good, we increase $\tilde p_j$ until either
$u_{jt_j}/\tilde p_j=\beta$, or $\tilde p_j=q_j$. In the latter case,
we buy the entire active segment of $j$, that is, we increase to
$y_{jt_j}=d_{jt_j}$. Thus, $t_j$ increases by 1. If we still have
$u_{jt_j}/q_j >\beta$, we again buy the entire active segment, and continue
until $u_{jt_j}/q_j\le \beta$ for the current active segment. This
finishes the description of the first stage.

\medskip

From the optimality conditions on $x$, it is easy to see that $y\le x$
at the end of the first stage. We claim that the following conditions
are satisfied at this point:
\begin{eqnarray}\label{eq:ba-compl}
y_{jt}=0\Rightarrow u_{jt}/\tilde p_j\le \beta,& &\quad
0<y_{jt}<d_{jt}\Rightarrow  u_{jt}/\tilde p_j=\beta,\quad y_{jt}=
d_{jt}\Rightarrow u_{jt}/\tilde p_j\ge \beta. \\
u(y)&\le &\min\{U,b\beta\} \label{eq:u-bound}\\
y_{jt}&>&c_{jt}\quad \Rightarrow\quad  \tilde p_j=q_j\label{eq:caps} 
\end{eqnarray}
The conditions \eqref{eq:ba-compl} and \eqref{eq:caps} are immediate from the
algorithm. The bound \eqref{eq:u-bound} follows since $y\le x$;
$u(y)\le u(x)\le U$ by the feasibility of $x$, and $u(x)\le b\beta$ by
the definition of $\beta$.

\paragraph{Stage II: price increases}
In the second stage we continue increasing $y$ and $\tilde p$, as well
as decreasing $\beta$ so that
\eqref{eq:ba-compl}, \eqref{eq:u-bound}, and \eqref{eq:caps} are maintained. The
algorithm terminates once \eqref{eq:u-bound} holds at equality. In
this case, one can verify from the KKT conditions that $y\in
G^u(\tilde p,b)$. Together with \eqref{eq:caps}, we see that the
output satisfies (A') and (B'). 
\medskip

The algorithm performs the following iterations. We let $A$ denote the
set of goods for which $u_{jt_j}/\tilde p_j=\beta$.
If there is a good $j\in A$ with $\tilde p_j=q_j$, then we start
increasing $y_{jt_j}$ until either 
	\begin{enumerate}
		\item $y_{jt_j} =d_{jt_j}$. Note that $t_j$ increases
                  by one in this case, and $j$ leaves $A$.
		\item The inequality \eqref{eq:u-bound} becomes
                  binding. In this case, the algorithm terminates.
	\end{enumerate}

We now turn to the case when
 $\tilde
p_j<q_j$ for all $j\in A$. During the iteration we multiplicatively increase the price of every good in
$A$ by the same factor $\alpha>0$, as well as decrease $\beta$ by the
factor $\alpha$. We choose the smallest value of $\alpha$ for which one of
the following events happens:

\begin{enumerate}
\item For some $j\in A$ we reach $\tilde p_j = q_j$. We change the
  allocations as described above.
\item The inequality \eqref{eq:u-bound} becomes
                  binding (due to the decrease in $\beta$). In this case, the algorithm terminates.
\item For some good $\ell\notin A$, $\frac{u_{\ell t_\ell}}{p_\ell} =
  \beta$. In this case, we add $\ell$ to $A$, and iterate with the
  larger set.
\end{enumerate}

It is easy to see that all three properties 
\eqref{eq:ba-compl}, \eqref{eq:u-bound}, and \eqref{eq:caps} are
maintained throughout the algorithm. We claim that the number of price
change steps is at most $\sum_j k_j$. Indeed, a price increase step
always ends when a good $j$ with $\tilde p_j=q_j$ enters $A$, either
in case 1 or case 3. Once this happens, we increase $y_{jt_j}$; if the
algorithm does not terminate, then we saturate the segment to
$y_{jt_j}=d_j$. This shows that the number of price augmentation steps
is bounded by the total number of segments $\sum_{j} k_j$.
\end{proof}

\begin{paragraph}{Bound on $p_{SR\max}$}
While the capped SPLC utilities are not strictly monotone 
nor differentiable, the same bound as in Lemma~\ref{lemma:pSRmaxBound} 
(or Lemma~\ref{lemma:pSRmaxBound-Gale}) can be similarly proved for
 $v_{i \max} = \max_{j \in [m], t\in [k_{ij}]} u_{ijt}$ and $v_{i \min} = \min_{j\in [m], t\in [k_{ij}]} \{u_{ijt} : u_{ijt} > 0\}$. The value $u_{ijt}$ represent the utility rate of agent $i$ for the $t$-th segment of good $j$.
\end{paragraph}

Recalling that $D_j$ is the number of available  units of good $j$, we have the following theorem.

\begin{theorem}
   Consider the Fisher market instance arising from the NSW problem where agents have capped SPLC utilities. 
   Let $K = \max_{i\in A} \sum_{j\in G} k_{ij}$ be the minimum number of segments needed to specify the utility of any agent.
   There is algorithm producing an $\epsilon$-SR equilibrium with respect to the Gale demand systems and bounds $t_j := D_j$ 
   that runs in time $O\left( \frac{n^3 m K}{\epsilon^2} \log\left(\frac{D_{\max} V_{\max}}{\epsilon} \right)\right)$.
\end{theorem}
\begin{proof}
  We start by adding a dummy agent $0$ to the market with budget $\epsilon / 5$.
  The utility of agent $0$ is linear, meaning that for each good $j$, there is only one segment of length $D_j$ and $u_{0, j , 1} = 1$.
  We initialize the auction algorithm by setting each price $p_j$ to $\frac{\epsilon}{5\cdot \sum_{j} D_j}$ and assigning all goods to $0$. 
  By running the auction algorithm for finding an SR-equilibrium we obtain $\frac{4\epsilon}{5}$-approximate equilibrium. 
  Now, we can remove the agent. 
  As this agent is buying the goods in amount at most $\epsilon/5$, 
  by removing the dummy agent we are left with a slightly weaker notion of $\epsilon$-approximate equilibrium.
  Namely, the first and third conditions of an approximate equilibrium (Definition~\ref{def:approx-SReq}) are satisfied by the choice of the precision parameter, but the second condition is not satisfied exactly.
  Rather, we can only guarantee that  $\sum_{i=1}^n x^{(i)}_j \le a_j$ and $\sum_{j\in [m]} p_j (a_j - \sum_{i=1}^n x^{(i)}_j) \le \epsilon/5$. 
  In words, the total price of unsold available amounts of all goods is at most $\epsilon/5$. 

  By Theorem~\ref{thm:SRrunning-oracle} the auction algorithm runs in
  $
  \displaystyle O\left(\frac{nmT_F}{\epsilon^2}\log \left(\frac{p_{SR\max}}{p_{\min}}\right)\right).
  $
  Recall that $T_F$ is time needed to implement \pFindNewPrices{}.
  By Lemma~\ref{lemma:findNewPricesBASPLC}, in this case $T_F$ is $O(K)$.
  By construction,  $p_{\min} = \frac{\epsilon}{\sum_{j} D_j}$.
  By Lemma~\ref{lemma:pSRmaxBound-Gale} (ii), $p_{SR\max} \le (1 + \epsilon)^n D_{\max} V_{\max}^{n-1}$.
\end{proof}

\subsection{Upper bound on the optimal NSW value}

Let $(x, p)$ an SR-equilibrium in the Fisher market arising from an instance of NSW where agents have Gale demand system corresponding to their utility functions and $e = (D_j)_{j\in[m]}$. In other words, $x_i \in G^{u_i}(p, 1)$ for each agent $i\in A$ and 
$\sum_{i\in[n]}x_{ij} = \sum_{i\in[n], j\in[m], t\in[k_{ij}]} x_{ijt} = D_j \cdot \min\{1, 1/p_j\}$ for each good $j \in G$.
As $x_i \in G^{u_i}(p,1)$ we have the following KKT conditions, see Section~\ref{section:basplcEquilibrium}:

\begin{enumerate}[label=(\roman*)]
\item $\frac{u_{ijt}}{u_i(x_i)} \le r_{ijt} +  p_j  + u_{ijt}\gamma_i$,
\item $\frac{u_{ijt}}{u_i(x_i)} = r_{ijt} + p_j +
  u_{ijt}\gamma_i$ whenever $x_{ijt} > 0$,
\item $x_{ijt}= d_{ijt}$ whenever $r_{ijt} > 0$, and
\item $\sum_j \sum_t x_{ijt} u_{ijt} = U_i$ whenever $\gamma_i>0$.
\end{enumerate}

Let us describe some properties of SR-equilibrium $(x,p)$ that the above KKT conditions imply.
By property (ii), $\frac{u_{ijt}}{r_{ijt} + p_j} = \frac{u_i(x_i)}{ 1 -  \gamma_i u_i(x_i)}$ whenever $x_{ijt}>0$.
This justifies defining 
\begin{equation}\label{eq:mbbi}
  {\textstyle \mbb_i} := \frac{u_i(x_i)}{ 1 -  \gamma_i u_i(x_i)}\,.
\end{equation}
Since the SR-equilibrium as well as NSW are invariant under scaling each agent's utilities $u_{ijt}$ and $U_i$, 
we can assume that $\mbb_i = 1$ for all agents $i$.
(This implies an appropriate implicit scaling of each $\gamma_i$ as well.)
By property (iii) we obtain:

\begin{proposition}\label{prop:MBB}
  If $x_{ijt} > 0$ then $\frac{u_{ijt}}{p_j} \ge 1$.
  If $\frac{u_{ijt}}{p_j} > \mbb_i = 1$ then $x_{ijt} = d_{ijt}$.
\end{proposition}
In other words, an agent only buys copies of goods with utility at least as much as their price, 
and if an agent values some copy of a good strictly more than its price then she also gets this copy in $x$.

We say that an agent $i$ is \emph{capped} if $u_i(x_i) = U_i$ and \emph{non-capped} otherwise.
Let $H(p) = \{j \in [m] : p_j > 1 \}$ be the set of \emph{expensive} goods.
\begin{proposition}\label{prop:uncapped}
Assume $\mbb_i=1$ for all agents $i$.
For each capped agent $i$, $j\in H(p)$, and $t\in[k_{ij}]$ we have $x_{ijt} = 0$ and $u_i(x)=U_i\le 1$. 
Each non-capped agent $i$ receives exactly one unit of utility, i.e., $u_i(x)=1$.
\end{proposition}
\begin{proof}
  Let $i$ be a capped agent.
  For a contradiction, suppose $x_{ijt} > 0$ for some $j\in H(p)$.
  Then $u_{ijt} \ge p_j > 1$. 
  Since $U_i \ge u_{ijt}$ it also holds that $U_i > 1$.
  This is a contradiction since $1 < \frac{U_i}{1-\gamma_iu_i(x_i)} = \frac{u_i(x_i)}{1-\gamma_iu_i(x_i)} = \mbb_i = 1$.

  Since $\frac{u_i(x_i)}{1-\gamma_iu_i(x_i)} = 1$ and $\gamma_i u_i(x_i) \ge 0$ it follows that $u_i(x_i) \le 1$.
  The fourth KKT condition implies that $\gamma_i = 0$ for non-capped agents, and therefore $u_i(x_i) = 1$.
\end{proof}

In order to prove an upper bound on the optimal NSW value we may assume that $U_i = \infty$ for all non-capped agents. 
Such an assumption can only increase the optimal NSW, so if we prove an upper bound under the assumption it also holds in the original instance. 
Since ``cap inequality'' ($\ldots \le U_i$) is ineffective for every non-capped agent, 
by the KKT condition $(x,p)$ remains an SR-equilibrium.
Let $A_c$ (resp. $A_u$) be the set of capped (resp. non-capped) agents in the equilibrium $(x, p)$.

\begin{lemma}
  Let $p$ be a vector of $SR$-equilibrium prices and $x^*$ an optimal NSW allocation. 
  Then 
  $$ \nsw(x^*) \le \left( \prod_{i\in A_c} U_i \cdot \prod_{j \in H(p)} p_j^{D_j} \right)^{1/n} \,.$$ 
\end{lemma}
\begin{proof}
First we give a bound on the sum of the agents’ utilities in any integer allocation $z$  as a function of prices $p$. 
Recall that $x$ is an SR-equilibrium allocation for prices $p$. 
Since utilities of the agents are scaled to have $\mbb_i =1$, by Proposition~\ref{prop:uncapped} each non-capped agent receives exactly $1$ unit
of utility in $x$.
Each capped agent receives $U_i$ utility in $x$ by definition. 
However, if there are some expensive goods then $x$ does not fully allocate all  the goods.
Each copy of the expensive goods generates $1$ unit of utility
in $x$ since the total spending on it each copy is precisely $1$ and since no capped agent buys expensive goods (Proposition~\ref{prop:uncapped}).

Let $\bar x$ be the allocation in which we allocate every copy of each expensive good $j$ to a single agent spending on it in $x$.
We can do so since the spending is exactly $D_j$ on copies of good $j$ and thus, there are at least as many agents buying good $j$ as the copies.  
As all of these agents are non-capped and we assume that for such agents $U_i =\infty$,
it follows that each copy of an expensive good generates exactly $p_j$ utility to the agents in $\bar x$. 
By Proposition~\ref{prop:MBB}, it is at least $p_j$  as $x_{ijt} > 0$ implies that $u_{ijt} \ge p_j$;
it is at most $p_j$ by the contrapositive of: $u_{ijt} > p_j$ implies that $x_{ijt} = d_{ijt} \ge 1$.
Therefore, the total utility that all the goods in $\bar x$ generate is:
$$
  \sum_{i \in A_c} U_i + |A_u|
  +\sum_{j\in H(p)} D_j(p_j-1) 
  = \sum_{i \in A_c} U_i + |A_u| - \sum_{j \in H(p)} D_j + \sum_{j\in H(p)} D_jp_j\,.
$$
The first two terms give the utility from $x$. Using $\bar x$ causes the utility to increase by $p_j - 1$ for each copy of good $j$ in
$H(p)$.

We claim that the total utility of all the agents in any \b{integer} allocation is not larger than the above sum.
Consider the copies of good $j$. 
In $\bar x$, each one of those goods generates either $p_j$ or more than $p_j$ utility. 
Moreover, any agent that can derive more than $p_j$ utility from a copy of a good actually
receives the copy in $\bar x$.
Therefore, $\bar x$ allocates the copy of goods to the agents such that the total utility all
the goods generate is maximized. 
It follows that for any integral allocation $z$ the total utility all agents receive is at most
$$
\sum_{i \in [n]} u_i(z) \le \sum_{i \in A_c} U_i + |A_u| - \sum_{j\in H(p)} D_j  + \sum_{j\in H(p)}  D_jp_j\,.
$$
At this point, suppose that we are given the above amount of utility and we can freely distribute it among agents to maximize NSW, 
regardless of what the utility function of each agent is, 
but only respecting the fact that the capped agents cannot get more than their cap,
and that expensive goods are indivisible. 
By Proposition~\ref{prop:uncapped}, all caps of the capped agents are at most $1$.
Then, it is not too hard to see that the optimal way of distributing our lump sum of utility is to assign: 
each expensive copy to a non-capped agent and nothing else to those agents, 
exactly $U_i$ to each capped agent, 
and $1$ to everyone else. 
In this case, the NSW is exactly $\left( \prod_{i\in A_c} U_i \cdot \prod_{j \in H(p)} p_j^{D_j} \right)^{1/n}$.
\end{proof}

\subsection{Rounding}
As in the previous section, we assume that the utilities are scaled such that $\mbb_i = 1$, see~\eqref{eq:mbbi}. 
Moreover, we use that $u_{ijt} \le U_i$.
We reduce our rounding to the case of capped linear utilities in~\cite{garg2018approximating}.
It is convenient to present the rounding in terms of the \emph{spending graph}.
For an SR-equilibrium $(x,p)$ the spending graph is a bipartite graph $(A, G; E)$
where an agent $i$ is adjacent to a good $j$ if and only if $x_{ij} > 0$.
We show how to round $x$ to an integral allocation $x'$. 

By the KKT conditions, whenever $\frac{u_{ijt}}{p_j} > \mbb_i$ then $x_{ijt}=d_{ijt}$ -- 
in this case we allocate $d_{ijt}$ copies of good $j$ to $i$ by setting $x'_{ijt} \leftarrow d_{ijt}$.
Moreover, if for some triple $i,j, t$ we have $x_{ijt} > 1$ then we allocate $\left \lfloor{x_{ijt}}\right \rfloor 
$ units of good $j$ to agent $i$. 
Formally, we set $x'_{ijt} \leftarrow  \left \lfloor{x_{ijt}}\right \rfloor $.
Once we do this for all goods and all agents, 
any agent can have at most one unit of each good that she is buying in the SR-equilibrium but that is not yet allocated in $x'$.
Hence there are at most $n$ units of each good $j$ that are still to be allocated.  
By the first rule for allocating goods, for these remaining copies of a good $j$, 
if an agent $i$ is buying a fraction of it, then $\frac{u_{ij t_i}}{p_j} = 1$ (where $t_i$ is the first non-saturated segment of agent $i$).
By assuming that $u_{ijt} = 0$ for all $t > t_i$, we can transform the instance into an instance 
in which the utility of every agent is capped linear. 
The only issue is that we could have several copies of a good.
Since there are at most $n$ copies of each good that are unassigned and the utilities are capped linear, 
we can simply split each good into the appropriate number of goods with a single copy.  
Then, the rest of the rounding follows the exact same steps as the rounding for capped linear utilities in~\cite{garg2018approximating}.
The analysis reduces in the same way. 
By choosing a suitable $\epsilon$ we obtain the following theorem.
%{\color{red}  It would be helpful to summarize these steps.}

\begin{theorem}
   Consider an instance of NSW problem where agents have  capped SPLC utilities.  
   Let $K = \max_{i\in A} \sum_{j\in G} k_{ij}$ be the minimum number of segments needed to specify the utility of any agent.
   Then there is an algorithm running in time $O\left(n^3 m K \log\left(D_{\max} V_{\max} \right)\right)$
   which produces a solution that is at most $2.404$ times worse than the optimum.
\end{theorem}

%%% Local Variables:
%%% mode: latex
%%% TeX-master: "main"
%%% End:

%% file: runtime.tex
%!TEX root = mainFileMOR.tex
\section{Running times of previous auction algorithms}
\label{section:running-times}
We review the running time bounds given in previous auction algorithms
and compare them to our bounds. We let $\1$ denote
the $n$ dimensional vector with all entries $1$.  
\paragraph{Linear utility functions, Garg and Kapoor~\cite{GargK06}}The paper includes
two algorithms. The running times are 
$O \Big(
\frac{nm}{\epsilon^2}\cdot$  
$\log\left(\frac{p_{\max} \cdot \pr{\1}{e}}{\epsilon \cdot  p_{\min}\cdot  e_{\min}}\right)
\cdot
\log\left(\frac{p_{\max}}{p_{\min}}\right) \Big)$ and 
$ O\left(\frac{nm}{\epsilon}  (n + m)
\log\left(\frac{p_{\max}}{p_{\min}}\right)  
\right)$, respectively.
The running time in Theorem~\ref{thm:running-oracle}, with the bound
$T_F=O(m)$ for linear utilities from Lemma~\ref{lemma:linearPrices},
gives an additional factor {$m$ when compared to the first bound}, while removing the first log term. 
We note that we are using a weaker notion of equilibrium in our result. 
The additional factor
is due to our global update step: due to the more general,
nonseparable nature of our framework, we consider all goods when updating
an agent, while the paper~\cite{GargK06} considers only one good for an
update. 

\paragraph{Separable WGS utilities, Garg, Kapoor, and Vazirani~\cite{garg2004auction}} The running time
bound is presented only for the Fisher market case, given as 
$O \left(\frac{nm}{\epsilon}
\log\frac{1}{\epsilon}  
\log \frac{v_{\max} \cdot \pr{\1}{b}}{ b_{\min} v_{\min}} 
\log m \right)$.
Here, $v_{\max} := \max_{i} v_{i\max}$ and
$v_{\min} := \min_{i} v_{i\min}$ are upper and lower bounds on the slopes of the utility functions,
  $b_{\min}$ is
the smallest budget, and $v$ is the total utility an agent would get
from owning the full amount of all goods. 
An issue with such a bound is that the value $\frac{v_{\max}}{v_{\min}}$ is not scale invariant.
Namely, the equilibrium in Fisher market remains the same even if each agent $i$ multiplies their utility function by a positive constant $\alpha_i$; but this changes the value $\frac{v_{\max}}{v_{\min}}$ arbitrarily. 
It is mentioned that the
result could be extended to exchange markets, similarly as in
\cite{GargK06}, but no details or running time estimation are provided.
\paragraph{Uniformly separable WGS utilities, Garg and Kapoor~\cite{garg2007market}} 
The paper gives essentially the same bound as in the case of separable WGS; 
the analysis is limited and mainly refers to the analysis of the auction algorithm for separable WGS utilities~\cite{garg2004auction}. 
A problematic issue is that the main motivation for the paper is to give bounds for CES and Cobb-Douglas utilities, 
but $v_{\max}=\infty$ for these particular utilities. 
%%% Local Variables:
%%% mode: latex
%%% TeX-master: "main"
%%% End:

%% file: CobbDouglasAgent.tex
%!TEX root = mainPlain.tex
\section{Adding a dummy agent to bound the prices}\label{section:dummy}
Using the construction in the papers~\cite{codenotti2005market, codenotti2005polynomial}, 
we present a general technique of modifying the market in order to bound $p_{\max}/p_{\min}$ in Theorem~\ref{thm:running-oracle}.
Given an exchange market $M$ with agents $A$ and goods $G$, 
we transform it to another market $\hat M$ with $n+1$ agents as follows.
Let $\eta \le 1$ be a parameter such that $\frac{\eta}{1+\eta} > \epsilon (1+\epsilon) m$ and $\epsilon (1+\epsilon) m \le 1/2$. 
For $i \in A$ we keep the same demand systems $D_i$ and the same initial endowments $e^{(i)}$. 
The market $\hat M$ has an extra agent $n+1$ with initial endowment 
$e^{(n+1)} = \eta e$  (recall $e = \sum_{i \in A} e^{(i)}$) and 
whose demand bundle is given via the Cobb-Douglas utility function $\left(\prod_j {x^{(n+1)}_j}\right)^{1/m}$.
Agent $n+1$ spends exactly $\frac{1}{m}$ 
of its budget on each good $j$ since its unique demand bundle $x^{(n+1)}$
is given by $ x^{(n+1)}_j = \frac{\eta \pr{p}{e}}{mp_j}$. 
 
The lemma below shows that adding such an agent can be used to bound
$\frac{p_{\max}}{p_{\min}}$, at the expense of working on a modified market.

\begin{lemma}\label{lemma:exchangeTransformation}
\begin{enumerate}[label=(\roman*)]
\item For an $\epsilon$-equilibrium of $\hat M$ formed by prices $p$ and bundles $x^{(i)}$, we have 
	$\displaystyle \frac{p_{\max}}{p_{\min}} \le \frac{(1+\epsilon)m}{\eta - \epsilon m (1+\epsilon)(1+\eta)}\cdot \frac{e_{\max}}{e_{\min}}$,
	where $e_{\max} = \max_{j} e_j$ and $e_{\min}=\min_{j} e_j$. 
\item An $\epsilon$-equilibrium in $\hat M$ gives an $\epsilon (1+\eta)$-equilibrium in $M$. 
\end{enumerate}
\end{lemma}

\begin{proof}
Consider an $\epsilon$-equilibrium in $\hat M$ formed by $p$ and bundles $x^{(i)}$.
By definition, there exist prices $p\le p^{(n+1)} \le (1+\epsilon)p$ and bundle $z^{(n+1)} \in D_{n+1}\left(p^{(n+1)}, \eta \pr{p}{e} \right)$ such that $x^{(n+1)}\le z^{(n+1)}$.
We have $z^{(n+1)}_j =  \frac{\eta \pr{p}{e} }{mp^{(n+1)}_j}$, and
therefore, $p_j z^{(n+1)}_j \ge \frac{\eta}{(1+\epsilon)m} {\pr{p}{e}}$.
On the other hand, from the third condition of the definition of $\epsilon$-equilibrium it follows that $p_j ( z^{(n+1)}_j - x^{(n+1)}_j) \le \epsilon \pr{p}{e} \cdot (1+\eta)$. 
Hence, $p_j x^{(n+1)}_j \ge \left( \frac{ \eta}{(1+\epsilon)m} - \epsilon (1+\eta)  \right) \cdot \pr{p}{e}$ for all $j$.
In particular, $x^{(n+1)}_j \ge \left( \frac{ \eta}{(1+\epsilon)m} - \epsilon (1+\eta) \right) \frac{p_{\max}e_{\min}}{p_j}$ for all $j$. Since $x^{(n+1)}_j \le e_j \le e_{\max}$  in an $\epsilon$-equilibrium, we have
$$
\frac{p_{\max}}{p_{\min}} \le  \left(\frac{ \eta}{(1+\epsilon)m} - \epsilon (1+\eta) \right)^{-1} \frac{e_{\max}}{e_{\min}} \,.
$$ 
The second part of the lemma follows easily from the definition of an approximate equilibrium.
\end{proof}